\documentclass[10pt,british]{llncs}
\pdfoutput=1
\usepackage{booktabs}
\usepackage[toc,page,header]{appendix}
\usepackage{etoc}
\usepackage[tight]{minitoc}
\usepackage{titletoc}

\usepackage{libertine}
\usepackage[T1]{fontenc}   
\usepackage{anyfontsize}

\usepackage{llncs-space}

\usepackage[dvipsnames]{xcolor}

\usepackage[defblank]{paralist}
\usepackage{enumitem}[shortlabels]

 \usepackage{cite}
\usepackage{varwidth}
\usepackage{xspace}

\usepackage{amsmath}
\usepackage{amssymb}
\usepackage{amsfonts}
\usepackage{csquotes}
\usepackage{algorithmicx}
\usepackage[ruled,vlined]{algorithm2e}

\usepackage{amsthm}

\DeclareMathSizes{10}{9}{7}{7}
\makeatletter
\g@addto@macro \normalsize {%
 \setlength\abovedisplayskip{5pt plus 2pt minus 2pt}%
 \setlength\belowdisplayskip{5pt plus 2pt minus 2pt}%
}
\makeatother

\usepackage{hyperref}

\newcommand\myshade{85}
\colorlet{mylinkcolor}{JungleGreen} % Aquamarine
\colorlet{mycitecolor}{YellowOrange}
\colorlet{myurlcolor}{violet}
\hypersetup{
  linkcolor  = mylinkcolor!\myshade!black,
  citecolor  = mycitecolor!\myshade!black,
  urlcolor   = myurlcolor!\myshade!black,
  colorlinks = true,
  linktoc  = page, % http://tex.stackexchange.com/questions/99130/coloring-the-page-numbers-in-the-table-of-contents
  bookmarksopen=true % To have it unfolded, just add this to the preamble of your document
}
\makeatletter
\newsavebox\myboxA
\newsavebox\myboxB
\newlength\mylenA

\newcommand*\xoverline[2][0.75]{%
    \sbox{\myboxA}{$\m@th#2$}%
    \setbox\myboxB\null% Phantom box
    \ht\myboxB=\ht\myboxA%
    \dp\myboxB=\dp\myboxA%
    \wd\myboxB=#1\wd\myboxA% Scale phantom
    \sbox\myboxB{$\m@th\overline{\copy\myboxB}$}%  Overlined phantom
    \setlength\mylenA{\the\wd\myboxA}%   calc width diff
    \addtolength\mylenA{-\the\wd\myboxB}%
    \ifdim\wd\myboxB<\wd\myboxA%
       \rlap{\hskip 0.5\mylenA\usebox\myboxB}{\usebox\myboxA}%
    \else
        \hskip -0.5\mylenA\rlap{\usebox\myboxA}{\hskip 0.5\mylenA\usebox\myboxB}%
    \fi}
\makeatother
\makeatletter
\def\thm@space@setup{%
  \thm@preskip=1.4ex 
  \thm@postskip=1.4\thm@preskip % or whatever, if you don't want them to be equal
}
\makeatother

\newcommand{\edge}{\textit{Eg}}

\newcounter{myexample}[section]

\newcounter{firstexample}[section]

\makeatletter
\renewcommand\paragraph{%
   \@startsection{paragraph}{4}{0mm}%
      {.25\baselineskip}%
      {-\baselineskip}%
      {\bfseries}}
\makeatother

\makeatletter
\def\old@comma{,}
\catcode`\,=13
\def,{%
  \ifmmode%
    \old@comma\discretionary{}{}{}%
  \else%
    \old@comma%
  \fi%
}
\makeatother

\newcommand{\observations}[1]{\subsubsection*{Observations:}}
\newcommand{\vocabulary}[1]{\subsubsection*{Vocabulary and Symbols:}}
\newcommand{\rules}[1]{\subsubsection*{Rules:}}
\newcommand{\relations}[1]{\subsubsection*{Relations:}}
\newcommand{\programname}[1]{\subsubsection*{Program Name:}}

\newcommand{\tX}{\xoverline{X}}

\newcommand{\tuple}[1]{\bar{#1}}
\newcommand{\tu}{\tuple{u}}

\newcommand{\tconst}{\bar c}   % vector of terms c

\newcommand{\blankk}
    {\raisebox{0.8ex}{-}\kern-0.4em b} % another blank ?

\newcommand{\col}{\colon}

\newcommand{\quotes}[1]{\lq\lq#1\rq\rq}         % makes ``#1''
\newcommand{\qif}{\leftarrow}                   % query if

\newcommand{\true}{\textit{true}\xspace}

\newcommand{\tpl}[1]{\bar{#1}}				%make tuple
\newcommand{\query}[2]{#1 \qif #2} 		%$\query{Q(\tpl x)}{B_1, .., B_n}$.
\newcommand{\eat}[1]{}
\newcommand{\A}{\mathcal{A}} 
\newcommand{\C}{\mathbf{C}} 
\newcommand{\D}{\mathcal{D}}

\renewcommand{\O}{\mathcal{O}} 
 
\renewcommand{\S}{\mathcal{S}}

\newcommand{\defterm}[1]{\mbox{\underline{\it\smash{#1}\vphantom{\lower.1ex\hbox
{x}}}}}

\newcommand{\la}{\leftarrow{}}

\newcommand{\incl}{\subseteq}

\newcommand{\set}[1]{\{#1\}}                      % set

\newcommand{\tup}[1]{\langle #1\rangle}            % tuple

\newcommand{\DP}{\textrm{DP}}

\newcommand{\NP}{\textrm{NP}\xspace}
\newcommand{\coNP}{\textsc{co-NP}\xspace}

\newcommand{\wrt}{wrt~}

\newcommand{\constantmath}[1]{\mathit{#1}}

\newcommand{\pupil}{\textit{pupil\/}}

\newcommand{\ccode}{\constantmath{code\/}}

\newcommand{\class}{\textit{class\/}}

\newcommand{\john}{\constantmath{john}}

\newcommand{\learns}{\textit{learns\/}}

\newcommand{\MCS}{\textsl{MCS}\xspace}

\newcommand{\kMCS}{\textsl{k-}\MCS\xspace}

\newcommand{\Compl}{\textsl{Compl}}
\newcommand{\prologcode}{%
http://www.inf.unibz.it/~savkovic/websvn/filedetails.php?repname=MAGIK-SVN&path=\%2FMAGIK_WEB_3\%2Fmagik_libs\%2Fprolog\%2Fmeta.pl}

\newcommand{\prologinput}{%
http://www.inf.unibz.it/~savkovic/websvn/filedetails.php?repname=MAGIK-SVN&path=\%2FMAGIK_WEB_3\%2Fmagik_libs\%2Fprolog\%2Finput.pl}

 \newcommand{\ei}{\emph{(i)\/}\xspace}
 \newcommand{\eii}{\emph{(ii)\/}\xspace}
 \newcommand{\eiii}{\emph{(iii)\/}\xspace}

\newcommand{\Gen}[1]{G_{#1}}				% generalization operator
\newcommand{\GenC}{\Gen{\C}}				% generalization operator for given C
\newcommand{\qcont}{\sqsubseteq}

\newcommand{\avl}[1]{{#1}^{a}}                           % available
\newcommand{\idl}[1]{{#1}^{i}}                           % ideal

\newcommand{\compl}[1]{\textsl{Compl}(#1)}               % generic completeness statement
\newcommand{\tc}[2]{\compl{#1;\,#2}}                     % TC statement, uses notation for \compl

\newcommand{\qc}[1]{\compl{#1}}                          % QC statement, uses notation for \compl
\newcommand{\cdb}[1]{\mathit{D}_{#1}}                    % canonical database 
\newcommand{\IFF}{\Longleftrightarrow}                   % if and only if arrow

\newcommand{\IDB}{\D}                                    % incomplete database
\newcommand{\di}{\idl D}                                 % ideal instance
\newcommand{\da}{\avl D}                                 % available instance

\newcommand{\drule}[2]{{#1}\qif{#2}}                     % datalog rule

\newcommand{\PIPTWO}{\Pi^\text{P}_2}
\newcommand{\school}{\textit{school\/}\xspace}
\newcommand{\sname}{\constantmath{sname\/}\xspace}

\newcommand{\pname}{\constantmath{pname\/}\xspace}

\newcommand{\type}{\constantmath{type\/}\xspace}

\newcommand{\district}{\constantmath{district\/}\xspace}
\newcommand{\lang}{\constantmath{lang\/}\xspace}

\newcommand{\goethe}{\constantmath{goethe\/}\xspace}

\newcommand{\conn}{\textit{conn\/}\xspace}

\newcommand{\newton}{\constantmath{\goethe\/}\xspace}
\newcommand{\halfDay}{\constantmath{halfDay\/}\xspace}
\newcommand{\fullDay}{\constantmath{fullDay\/}\xspace}
\newcommand{\QC}{Q_{C}}			% canonical query
\newcommand{\english}{\constantmath{english\/}\xspace}	% in formula
\newcommand{\primary}{\constantmath{primary\/}\xspace}

\newcommand{\merano}{\constantmath{merano\/}\xspace}

    \newcommand{\indexpb}{\textit{pb}} % pupils bolzano, 
\newcommand{\indexsp}{\textit{sp}} % primary schools
\newcommand{\indexenp}{\textit{enp}} % English primary 
\newcommand{\indexppb}{\textit{ppb}} % primary pupils bolzano
\newcommand{\indexpbl}{\textit{pbl}} % primary pupils bolzano learns
\newcommand{\indexl}{\textit{l}} % primary pupils bolzano learns first
\newcounter{firstexampleamw}[section]

\newcounter{examplectdamw}[section]

\newcommand{\dom}{{\sf dom\/}\xspace}% avail constants for APs
\DeclareMathAlphabet\mathbfcal{OMS}{cmsy}{b}{n}

\newcommand{\nn}{\mathbb{N}}

\newcommand{\te}[1]{\text{~#1~}}

\newcommand{\mgu}{\text{mgu}}

\newtheorem{theoremx}[theorem]{Theorem}
\newtheorem{propositionx}[theorem]{Proposition}
\newtheorem{lemmax}[theorem]{Lemma}
\theoremstyle{definition}
\newtheorem{definitionx}[theorem]{Definition}
\newtheorem{examplex}[theorem]{Example}
\newtheorem*{continuedex}{Example \continuedexref\space Continued}

\newcommand{\myqedsymbol}{\ensuremath{\triangle}}%

\newenvironment{examcont-new}[1]
  {\newcommand{\continuedexref}{\ref*{#1}}\continuedex}
  {\hfill\myqedsymbol\endcontinuedex}

\newenvironment{example-new}
  {\pushQED{\qed}\examplex}
  {\popQED\endexamplex}

 \newcommand{\test}{\mathit{test}}

\newcommand{\SUB}[1]{\textsl{Sub}_{#1}}
\begin{document}

\title{Complete Approximations of Incomplete Queries}
\titlerunning{}  
% abbreviated title (for running head)

\newcommand{\authspace}{\hspace{1.5ex}}
\author{%
  % ~
Julien Corman \and 
Werner Nutt \and 
Ognjen Savkovi\'c 
% Ognjen Savkovi\'c$^{1}$\authspace%
% Evgeny Kharlamov$^{2}$\authspace%
% Werner Nutt$^{1}$\authspace%
% Pierre Senellart$^{3}$%
}		
\institute{
%1
	Free University of Bozen-Bolzano, Italy 
}

%----------------------------------------------%            
 % typeset the title of the contribution

%----------------------------------------------%
% CONTENT
%----------------------------------------------%

\maketitle

% \renewcommand\qedsymbol{$\blacksquare$}

% !TEX root = ../rr-2024-main.tex
\begin{abstract}
This paper studies the completeness of conjunctive queries over a partially complete database and the approximation of incomplete queries.
Given a query and a set of completeness rules (a special kind of tuple generating dependencies) that specify which parts of the database are complete,
we investigate whether the query can be fully answered, as if all data were available.
If not, we explore reformulating the query into either Maximal Complete Specializations (MCSs) or the (unique up to equivalence) Minimal Complete Generalization (MCG) that can be fully answered, that is,
the best complete approximations of the query from below or above in the sense of query containment.
We show that the MSG can be characterized as the least fixed-point of a monotonic operator in a preorder.
Then, we show that an MCS can be computed by recursive backward application of completeness rules.
We study the complexity of both problems and discuss implementation techniques that rely on an ASP and Prolog engines, respectively.
\end{abstract}

% !TEX root = ../rr-2024-main.tex

\section{Introduction}
%%%%%%%%%%%%%%%%%%%%%%%%%%%%%%%%%%%%%%%%%%%%%%%%%%%%%%%%%%%%%%%%

Completeness is one of the classical dimensions of data quality, well studied both in the context of relational data~\cite{%
Fan:Geerts-relative_information_completeness:pods:09,%
Fan:Geerts-capturing_missing_tuples_and_values:pods:10,%
Libkin:SQL:Null:PODS2016,Arenas:Aggregates:IncosistentDBs:Journal:2003
}
and more recently in the context of knowledge graphs~\cite{%
  Razniewski-2024-Comple-Survey,%
  fensel2020knowledge}.
An information system may be incomplete because some data has not yet been inserted,
or simply because some information is unavailable for certain records.
As a consequence, a query $Q$ that retrieves data may itself be incomplete,
in the sense that when $Q$ is evaluated over the actual (incomplete) database, 
some answers to $Q$ over the ideal (complete) database are missing.
Similarly, if $Q$ produces statistics (e.g., count answers),
then these statistics may be incorrect.

On the other hand, part of the data is often \emph{known} to be complete.
This information may be sufficient to confirm that $Q$ is complete, or,
if this is not the case,
to produce a query that is complete and \emph{approximates} $Q$.
The latter task is the one investigated in this paper.
More precisely, we consider two natural strategies to approximate $Q$.
The first strategy consists in producing a complete query $Q'$ that is more \emph{general} than $Q$,
in the sense that the answers to $Q$ over the ideal (complete) database are a subset of the answers to $Q'$.
In a scenario where one is searching for a specific piece of information,
this strategy ensures that it will be retrieved.
Instead, the second strategy consists in producing a complete query that is more \emph{specific} than $Q$.
This can be useful for statistics:
the answers to $Q$ may not be complete yet (e.g., for a whole country),
but one may already publish partial statistics (e.g., for a given region) that are guaranteed to be correct.

%   In particular, in information systems where data may arrive at different speeds,
% analytical computations may be affected by incomplete data.
% Hence, it is desirable to understand whether query answers can be trusted if data are incomplete.
% If not, how much can those answers differ from factual answers?
% One way to verify this is by evaluating similar queries, whose answers are guaranteed to be complete.

As a (toy) running example, we will consider a database instance $D$ of a hypothetical school information system in the Italian province of Bolzano.
% where stakeholders require assurances about the completeness of query answers.
% In this scenario, a database instance $D$ may be \emph{partially complete} by including,
% for instance, all pupils in schools within the district of Bolzano,
% while being \emph{generally incomplete\/} because the registration data from other schools in Bolzano and other district in the province have not yet been incorporated.
%
Following the approach of~\cite{levy_completeness,Razniewski:Nutt-Compl:of:Queries-VLDB11,CIKM2015-Nutt},
one can express which (parts of which) tables in $D$ are known to be complete, using so-called \emph{table-completeness statements} (TCSs).
For example, one may write three TCSs that say
that $D$ 
% the database instance 
contains 
\ei \quotes{all primary schools,}
% ({\emph{prim}}),
\eii \quotes{all pupils attending a school in the district of Merano,} and
% ({\emph{schBol}}), 
\eiii \quotes{all English learners that are pupils at a primary school,}
% ({\emph{primEng}}),
respectively. 
% Such TC statements may be produced by workflow engines or during ETL processes that populate data warehouses.
Conceptually, the TCS \eii for instance asserts that for every pupil enrolled in a school in the district of Merano,
there exists a corresponding record in $D$'s \quotes{pupils} table.

Now consider the query $Q$ that retrieves all language learners enrolled in a primary school in the district of Merano.
The three TCSs above are not sufficient to infer that $Q$ is complete.
However, one can approximate $Q$ by either generalizing or specializing it.
For instance, the query $Q^+$ that retrieves all pupils enrolled in a primary school in the district of Merano generalizes $Q$,
and is complete, from TCSs \ei and \eii.
If the purpose of $Q$ was to retrieve a specific pupil (among the ideal answers to $Q$),
then the answers to $Q'$ are guaranteed to contain this pupil.
Conversely,
the query $Q^-$ that retrieves all English learners enrolled in a primary school in the district of Merano specializes $Q$,
and is also complete, from TCSs \ei, \eii and \eiii.
If the purpose of $Q$ was to publish statistics (e.g., its number of answers),
then partial statistics (e.g., the number of answers to $Q$ restricted to English learners, a.k.a. $Q^-$) can already be published.

% If we know that these three TC statements (TCSs) hold over $D$,
% then we may infer that a certain query $Q$ will return the same answers over $D$ as it would
% over the ideal (complete) database instance.
% that contains a complete record of the application domain.
% For instance, 
% from from the TC statements ({\emph{prim}}), ({\emph{schBol}}), and ({\emph{primEng}})
% we can infer that the complete set of pupils that are enrolled in primary schools in Bolzano and learn English can be retrieved.
% ({\emph{primBolEng}}).

Previous work studied the problem of checking whether a conjunctive query (CQ) is complete wrt to a set of TCSs.
In particular, it was shown in~\cite{Razniewski:Nutt-Compl:of:Queries-VLDB11} that this problem can be reduced to the classical problem of query containment~\cite{Chandra-CQ-containment-77}.
Then in~\cite{CIKM2015-Nutt}, this same problem was reduced to reasoning in Datalog programs under answer-set semantics,
extending the setting by considering the presence of keys, foreign keys, and finite domain constraints.
We pursue this work by studying the two types of complete approximations (generalization and specialization) sketched above,
still focusing on CQs.

% A problem that has not been studied is how to \emph{approximate\/} an incomplete query with alternative complete queries.
% One way to do this is by searching for queries that are either more general or more specific.
% In our example, a query asking for all pupils may be incomplete since there are no guarantees for all pupils.
% However, ({\emph{primBolEng}}) is a similar query (in this case more special) that is complete, i.e., which answers can be trusted. 
% %
% Alternatively, in such information systems, both complete generalizations and specializations can be used
% to estimate the query answer bounds of certain analytical queries that are potentially incomplete (e.g., counting queries).

\paragraph*{Contributions}
% Formally, if $Q$, $Q'$ are queries such that $Q$ is contained in $Q'$, 
% then we say that $Q$ is a \emph{specialization\/} of $Q'$ and
% $Q'$ is a \emph{generalization\/} of $Q$.
% One can prove that, given TC statements $\C$, 
% if there is a complete generalization of a query $Q$, 
% there is a most specific one.
% It can be computed by a fixpoint iteration that drops \quotes{incomplete} 
% atoms from $Q$. 
% This work focuses on CQs.
Among the complete CQs that generalize a CQ $Q$, we identify the most specific ones,
which we call \emph{Minimal Complete Generalizations\/} (MCGs).
% as those that are the most specific among the complete ones.
And we define the \emph{Maximal Complete Specializations} (MCSs) of $Q$ analogously.
% onjunctive queries (CQs).
% and our theory can be extended to Union of CQs.
Our main contributions can be summarized as follows:
\begin{itemize}%[noitemsep,nolistsep]
\vspace*{-0,5em}

\item We show that if $Q$ admits an MCG, then it is unique up to query equivalence.
  To this end, we introduce a generalization operator that maps incomplete queries to \enquote{less} incomplete ones, 
  and show that a fixed-point of this operator (modulo query equivalence) is a complete generalization.
  This yields an algorithm to compute an MCG.
  We also study the computational complexity of the corresponding decision problem.
\item We show that there can be multiple MCSs, 
  or (infinitely many) complete specializations but no maximal one.
  As a consequence, we restrict our search space for an MCS to queries of a given bounded size.
  For this latter problem, we first propose an algorithm that finds MCSs (if any) that can be obtained without adding any atom to $Q$,
leveraging unification between the query under construction and atoms in TCSs.
  We then extend this idea to identify MCSs that may contain more atoms.
  \item 
  %Finally, we discuss implementation techniques for finding the MSG and MCSs. Since we characterize first problem with the forward rule application, and the second with the backward, we establish implementation techniques in Datalog and Prolog respectively.
  We discuss implementation techniques for finding MSG and MCSs using forward rule application in Datalog
  and backward rule application in Prolog, respectively.
\end{itemize}

% Due to space limitations, the paper contains only selected proofs and intuitions for theorems.
% A complete version can be found at:
% \url{XXX}
% As suggested in the call, complete proofs are provided in the supplementary material:
% \url{https://anonymous.4open.science/r/QueryApprox-946D/}

\paragraph*{Organization}

% The remainder of the paper is organized as follows. 
In Section~\ref{sec:prelim}, we recall basic definitions from database theory and
formally introduce the problem of query completeness and relevant characterization theorems.
In Sections~\ref{sec:query:generalization} and~\ref{sec:query:specialization},
we discuss the generalization and specialization problems, respectively.
Section~\ref{sec:implementation} discusses corresponding implementation techniques.
% and Section~\ref{sec:conclusion} concludes.

%\ognjen{fix MCG or MSG everywhere...}

% \input{input/2-1-running-example}
% !TEX root = ../rr-2024-main.tex

\section{Query Completeness}
\label{sec:prelim}
% \paragraph*{Query and Table Completeness}
% --------------------------------------%
\paragraph*{Preliminaries}
We adopt standard notation from databases.
We assume an infinite set each of relation symbols, constants, and variables.
The constants make up the set $\dom$ (the \emph{domain}).
Terms are constants or variables.
For variables, we use upper-case, for constants lower-case letters,
tuples are indicated by an overline, like $\tconst$.
%
% A \emph{database instance} is a finite set of ground atoms, called \emph{facts},
% over  a schema and the \emph{domain} $\domQPos = \dom \cup \RationalsPos$.
%
%
For a relation symbol $R$ with arity $n$, an \emph{atom\/} is an expression
$R(t_1,\ldots,t_n)$, where $t_1,\ldots,t_n$ are terms.
A database \emph{instance} $D$ is a finite set of ground atoms (facts),
that is, atoms that contain only constants.
% We sometimes refer to the atoms in a database instance as facts.
% For a relation $R$, we denote as $R(\D)$ the set 
% $\set{\tpl t \mid R(\tpl t)\in \D}$.
%of all tuples occurring in an $R$-atom in $\D$.
% A \emph{condition} is a set of atoms.
A \emph{conjunctive query\/} (CQ) is written as 
$\query{Q(\tpl u)}{B}$, where $Q(\tpl u)$ is an atom and $B$ a conjunction of atoms.
We call~$B$ the \emph{body\/} of $Q$ and $\tpl u$ the \emph{head terms.}
A CQ is \emph{safe\/} if all variables of $\tpl u$ occur in $B$.
Given a CQ $\query{Q(\tpl u)}{B}$ and an instance $D$, 
a tuple $\alpha\tpl u$ is an \emph{answer\/} of $Q$ over $D$, written $Q(D)$,
if
%  an answer to $Q$ is a tuple $\alpha\tpl u$, 
% where 
$\alpha$ is an assignment of variables from $Q$ to $\dom$ such that 
$\alpha B \subseteq D$. 
A query $Q$ is \emph{contained\/} in a query $Q'$, written $Q \qcont Q'$, 
if for every database instance $D$ it holds that $Q(D) \subseteq Q'(D)$. 
We say that $Q$ and $Q'$ are \emph{equivalent\/} if $Q \qcont Q'$ and $Q' \qcont Q$.
The query $Q$ is \emph{strictly contained\/} in $Q'$ if $Q$ is contained in $Q'$ and $Q'$ is not equivalent to $Q$.
If $Q$ is contained in $Q'$, then $Q'$ is \emph{more general\/} than $Q$,
and $Q$ is \emph{more special\/} than $Q'$.
% A query $Q'$ is a \emph{complete generalization\/} of $Q$ wrt 
% a set of TC statements $\C$ if $Q'$ is more general than $Q$
% and it is complete wrt $\C$.

% Werner to copy this paragraph
% A query $Q$ is a \emph{minimal} query under set semantics 
% if it does not have redundant atoms, that is, if for any (safe) query
% $Q'$ obtained from $Q$ by dropping some of the atoms in $Q$ it holds:
% $Q \not \sqsubseteq Q'$.
% %
% Let $Q(\tX)\la B$ be a query, and $\alpha$ a substitution,
% we use $\alpha Q$ to denote the query of the form 
% $Q'(\alpha \tX) \la \alpha B$.

\paragraph*{Completeness Theory}

When stating that data is potentially incomplete, one must have a conceptual complete reference. 
We model an \emph{incomplete database\/} in the style of \cite{motro_integrity} 
as a pair of database instances $\IDB=(\di,\da)$, where $\da\incl\di$.
Here, $\di$ is called the \emph{ideal\/} state and $\da$ the \emph{available\/} state.
In an application, the state stored in a DBMS is
the available state, which often represents only a part of the facts that hold in reality.
The facts holding in reality constitute the ideal state, which however is unknown.

Given a query $Q$ and an incomplete database $\IDB = (\di,\da)$,
we say that $Q$ is \emph{complete over~$D$\/}
if $Q(\di) = Q(\da)$ and we write $\D\models \qc Q$.

To specify that parts of a relation instance are complete, we introduce a kind of metadata,
which we call \emph{table completeness\/} (TC) statements (or TCS for short).
A~TCS, 
written $\tc{R(\tpl s)}G$, has two components,
a relational atom $R(\tpl s)$, where $\tpl s$ is a tuple of terms, 
and (a possibly empty) conjunction of atoms~$G$, called \emph{condition.} 
Intuitively, a TC statement $\tc{R(\tpl s)}G$ asserts that table $R$
is complete for all tuples that match $\tpl s$ and can be joined with $G$.
We denote a TC statement generically as $C$.

To define the semantics of a TCS $C$, 
we associate to it a query $\query{\QC(\tpl s)}{R(\tpl s), G}$.
Then $C$ is satisfied by 
$\IDB = (\di,\da)$, written $\IDB \models \tc{R(\tpl s)}{G}$,
if $\QC(\di) \incl R(\da)$.
This means that the ideal instance $\di$ is used to determine the tuples 
in the ideal instance $R(\di)$ of $R$ that match $\tpl s$ and satisfy~$G$, and that 
the statement $C$ is satisfied if these tuples are present in the available
version $R(\da)$. If no TC statement is associated with a relation then we do not know anything about the completeness of that relation. 

Let $\C$ be a set of TC statements.
We say that $\C$ entails the completeness of $Q$,
written $\C\models\qc Q$, if 
for every $\D$ it holds that $\D \models \C$ implies $\D\models \qc Q$.

\begin{example-new}
\label{ex:intro}
We consider a simplified database 
\quotes{schoolBolzano}, modeling schools in the province of Bolzano. 
It consists of the three tables below: %, where primary keys are underlined: 
\begin{align*}
& \pupil({\pname},\ccode, \sname),  & %\class({\ccode},{\sname},\level,\scheme),\\ 
& \school({\sname},\type,\district), & \learns({\pname},{\lang}).   
% \\
% % %
% & \takes(\underline{\pname},\underline{\activity}). 
\end{align*}
The table $\pupil$ contains for each pupil the name,
the class and the school the pupil attends.
%
% The table $\class$ stores for each class, which is identified by a code
% and a school name,
% the level (e.g., 1$^{\text{st}}$ year) and its scheme (e.g., half day or full day).
%
The table $\school$ records for each school the name, the type (e.g.,
primary or middle school), and its school district.
Finally, the table $\learns$ records which pupil learns which language. % and which takes some activity.
% Note that a pupil can learn more than one language.
%and attend more then one activity.
%
Next, we assume the following completeness statements:
\begin{align*}
C_\indexsp&\colon&{\!}&\Compl(\school(S,\primary,D);\,\true),\\
C_\indexpb&\colon&{\!}&\Compl(\pupil(N,C,S);\,\school(S,T,\merano)),\\
% C_\indexhd&\colon&{\!}&\Compl(\class(C,S,L,\halfDay);\true),\\
% C_\indexfd&\colon&{\!}&\Compl(\class(C,S,L,\fullDay);\true),\\
C_\indexenp&\colon&{\!}&\Compl(\learns(N,\english);\,\pupil(N,C,S),\school(S,\primary,D)). %,\\
%C_\indexenm&\colon&{\!}&\Compl(\learns(N,\english);\,\,\pupil(N,C,S),\school(S,\middleschool,D)).
\end{align*}
The statements say that the database contains
all primary schools ($C_\indexsp$),
all pupils attending a school in the district of Merano ($C_\indexpb$), and
all English learners that are pupils attending a primary school ($C_\indexenp$).

For instance, let $D^a=\set{\school(\goethe,\primary,\merano})$
and $D^i = D^a \cup \set{\pupil(\john,1,\goethe)}$. 
Then, $(D^i,D^a)\models C_\indexsp$, but 
$(D^i,D^a)\not \models C_\indexpb$.
%  since 
% $\pupil(\john,1,\goethe) \not \in D^a$.

Now consider the query for the \enquote{names of all pupils that attend a primary school in the district of Merano}:
\begin{align*} 
Q_\indexppb(N) \la
  \ \pupil(N,C,S),\school(S,\primary,\merano).
\end{align*}
Intuitively, the query is complete because our database contains all records of
pupils from primary schools in the district of Merano ($C_\indexpb$)
and it contains all such schools ($C_\indexsp$).

Alternatively, let us consider the query for \enquote{the names of all pupils that attend a primary school in the district of Merano and that learn some language}:
\begin{align*}
Q_\indexpbl(N)  \la
   \ \pupil(N,C,S), \school(S,\primary,\merano), \learns(N,L)
\end{align*}

This query is not complete since some language learners (of a language other than English) may be missing from the database. 
%'In such cases, we would like to find a complete query that is close to our incomplete query. 
\end{example-new}

% \note{W: Drop takes if not needed; O: takes is better for tests.}

\paragraph*{Reasoning about Query Completeness}
To reason about completeness, 
we define for every set $\C$ of TC statements the operator~$T_\C$ that maps database instances to database instances.  
If $C$ is a TC statement about $R$, then we define 
% $T_C(D) := \set{R(\tpl t) \mid \tpl t \in \QC(D)}$.
% For $\C$ we define 
% \begin{equation}
%   \label{eqn-TC:Operator}
% \mbox{$T_\C(D) := \bigcup_{C\in\C} T_C(D).$}
% \end{equation}
% Then we define
\begin{equation}
  \label{eqn-TC:Operator}
 \mbox{$T_C(D) := \set{R(\tpl t) \mid \tpl t \in \QC(D)}$} \quad\text{and}\quad\mbox{$T_\C(D) := \bigcup_{C\in\C} T_C(D).$}
\end{equation}
Since all queries $Q_C$ are monotonic, the operator $T_\C$ is monotonic in that
for any $D\incl D'$ it holds that $T_\C(D)\incl T_\C(D')$.
Moreover, for every instance~$D$, the pair $(D,T_\C(D))$ is an incomplete database that satisfies $\C$,
and $T_\C(D)$ is the smallest set (wrt set inclusion) for which this holds.
Below we cite a proposition %Proposition~1 in \cite{CIKM2015-Nutt}, which
that summarizes the connection between
the $T_\C$ operator and the satisfaction of the set $\C$ by a in incomplete database:

\begin{propositionx}[\mbox{${T_\C}$} Operator \protect{\cite[Proposition 1]{CIKM2015-Nutt}}]
\label{proposition-tc-operator}
Let $\C$ be a set of TC statements. Then
% \begin{itemize}
% 	\item[a)] $T_\C (D) \subseteq D$, \quad for all database instances $D$;
% 	\item[b)] $(D^i,D^a) \models \C$ \  iff\ \ $T_\C(D^i) \subseteq D^a$, \quad
% 	for all $D^a \subseteq D^i$.
% \end{itemize}
\begin{itemize}
  \vspace{-1ex}
	\item $T_\C (D) \subseteq D$, \quad for all database instances $D$;
	\item $(D^i,D^a) \models \C$ \  iff\ \ $T_\C(D^i) \subseteq D^a$, \quad
	for all $D^a \subseteq D^i$.
\end{itemize}
\end{propositionx}

In \cite{CIKM2015-Nutt} it has been shown that, 
similar to containment checking \cite{Chandra-CQ-containment-77}, 
completeness checking can be reduced to the question whether
a test query returns a specific result over a test database.
Below, we summarize that approach.

In what follows, we always consider a set of TC statements $\C$ and a conjunctive query~$Q$ defined by the rule
$\drule{Q(\tpl u)}{R_1(\tpl t_1),\ldots,R_n(\tpl t_n)}$.
We define the set of facts $\cdb Q$,
which we call the \emph{canonical database\/} of $Q$,
obtained by freezing the atoms in the body of $Q$.%
\footnote{%
\quotes{Freezing} variables is a well-known concept in logic programming and database theory,
which allows one to treat a variable like a constant.}
Thus, 
\begin{equation*}
  \label{eqn-canonical:db}
    \cdb Q = \set{R_1(\theta\tpl t_1),\ldots,R_n(\theta\tpl t_n)},  
\end{equation*}
where $\theta$~is the substitution that maps each variable $X$ to 
the \quotes{frozen version} $\theta X$ of $X$.

To check whether $Q$ is complete
one applies $T_\C$ to $\cdb Q$
and verifies whether $Q$ can retrieve the 
frozen tuple of distinguished variables.

\begin{theoremx}[Characterization of Completeness \protect{\cite[Theorem 2]{CIKM2015-Nutt}}]
 \label{theo-plain:reasoning}
% [from~\cite{CIKM2015-Nutt}, Theorem 2].
Let $\C$ be a set of TC statements, and 
$\query{Q(\tpl u)}{B}$ be a conjunctive query. Then
\begin{equation*}
  \label{characterization-plain:reasoning}
\C \models \qc Q
\quad\IFF\quad
\theta\tpl u \in Q(T_\C(\cdb Q)).
\end{equation*}
\end{theoremx}

\begin{example-new}
We illustrate this theorem with our running example.
Suppose $\IDB=(\di,\da)$ satisfies $\set{C_\indexsp,C_\indexpb}$. 
Consider an answer, say $n'$, returned by $Q_\indexppb$ over the ideal instance $\di$.
Then $\di$ contains two atoms of the form 
$\pupil(n',c',s')$ and 
$\school(s',\primary,\merano)$.
Now, due to $C_\indexsp$, also $\da$ contains $\school(s',\primary,\merano)$, and
due to $C_\indexpb$, the atom 
$\pupil(s',c',\newton)$ is in $\da$, too.
Consequently, $Q_\indexppb$ returns $n'$ also over $\da$.
Since $\di$ and $\da$ were arbitrary, this shows that 
$\set{C_\indexsp,C_\indexpb} \models \qc{Q_\indexppb}$.
\end{example-new}

\begin{example-new}[Complete Generalizations and Specializations]
We recall query $Q_\indexpbl(N)$ that was incomplete given the assumptions.
One way to obtain more general queries is by dropping some of the query atoms (that is, relaxing query conditions). 
Hence, if we drop the $\learns$ atom we obtain a complete generalization of $Q$: 
\[  
Q_\indexpbl^{\textit{gen}}(N)  \la \pupil(N,C,S),\school(S,\primary,\merano)
\]
which is the same as query $Q_\indexppb$ and is complete. 

To obtain query specializations one can turn variables into constants or  add new atoms to the body (or both).
For instance, we observe if in $Q_\indexpbl(N)$ we replace $\learns(N,L)$  with 
$\learns(N,\english)$ we obtain a query:
\begin{align*}
Q_\indexpbl^{\textit{spec}}(N)  \la
    & \ \pupil(N,C,S), \school(S,\primary,\merano), \learns(N,\english)
\end{align*}
which is complete, since $C_\indexenp$ guarantees completeness of the $\learns$-atom.
\end{example-new}

% !TEX root = ../rr-2024-main.tex

\section{Query Generalization}

\label{sec:query:generalization}

% In this section, we investigate the problem of query generalization. 
% %
% First we deal with plain case, and then we investigate cases with integrity constraints. 
% %
% For each of the cases we introduce algorithms that computes generalizations.

% %---------------------------------------------%
% %---------------------------------------------%
% \subsection{Query Generalization -- Plain Case}
% \label{sec:query:generalization:plain}

% A query $Q$ is \emph{contained\/} in a query $Q'$, written $Q \qcont Q'$, 
% if for every database instance $D$ it holds that $Q(D) \subseteq Q'(D)$. 
% We say that $Q$ and $Q'$ are \emph{equivalent\/} if $Q \qcont Q'$ and $Q' \qcont Q$.
% The query $Q$ is \emph{strictly contained\/} in $Q'$ if $Q$ is contained in $Q'$ and $Q'$ is not equivalent to~$Q$.
% If $Q$ is contained in $Q'$, then $Q'$ is \emph{more general} than $Q$.

% Containment establishes a preorder on conjunctive queries 
In this section, we study the problem of finding a complete generalization of a (possibly incomplete) query.
We restrict our scope to the case where all queries (input query and generalizations) are CQs. 

First, we rephrase the well-known characterization of containment among conjunctive queries
due to Chandra and Merlin \cite{Chandra-CQ-containment-77} in our formalism.

\begin{propositionx}[Characterization of Containment]
  \label{prop-containment}
Let $\query{Q(\tpl u)}{B}$ and $\query{Q'({\tpl u}')}{B'}$ be conjunctive queries. Then
\begin{equation*}
  \label{characterization-containment}
  Q \qcont Q'
\quad\IFF\quad
\theta\tpl u \in Q'(\cdb Q).
\end{equation*}
\end{propositionx}
The proposition holds because $Q$ is contained in $Q'$ iff there exists a homomorphism~$\delta$ from $Q'$ to $Q$.
Such a $\delta$ can be seen as an assignment of constants and variables in $D_Q$ to the variables in $Q'$
that satisfies $Q'$ and returns the tuple of head terms $\theta\tpl u$ of $Q$.

We say that a query $Q'$ is a \emph{complete generalization\/} of $Q$ wrt 
a set of TC statements $\C$ if $Q'$ is more general than $Q$
and $Q'$ is complete wrt $\C$.
\begin{definitionx}[Minimal Complete Generalization]
A query $Q'$ is a \emph{minimal complete generalization\/} (MCG) of $Q$ 
wrt a set of TC statements $\C$ if:  
\begin{itemize}
\item $Q'$ is a complete generalization of $Q$ wrt $\C$, and
\item it is a minimal one, that is, there exists no query $Q''$ such that 
     $Q''$ is a complete generalization of $Q$ and
     $Q''$ is strictly contained in $Q'$.
\end{itemize}
\end{definitionx}

%--------------------------------------------------------%
\paragraph*{Computing Minimal Complete Generalizations}
%--------------------------------------------------------%

In the following we investigate  properties of MCGs.
% Several questions arise from the definition:
% % \note{Do I want this questions; Or in this form?}
% \begin{enumerate}
%   \item[\ei] How many MCGs can a query have (at most)?
%   \item[\eii] How can we design an algorithm that, given $Q$ and $\C$, 
%               computes one (or all) MCGs of $Q$?
%   \item[\eiii] Given queries $Q$, $Q'$, and a set of TC statements $\C$,
%      how complex is it to decide whether $Q'$ is an MCG of $Q$ wrt $\C$?
% \end{enumerate}
%
We first show that in our search for MCGs of $Q$ it is sufficient to concentrate on \emph{subqueries\/} of $Q$,
that is, (safe) queries obtained by dropping some atoms from the body of $Q$.
Note that if $Q_0$ is a subquery of $Q$, then $Q$ is contained in $Q_0$.
More specifically, we will show the following proposition.

\begin{propositionx}[MCGs are Subqueries]
  \label{prop-MCGs:Are:Subqueries}
  Let $\C$ be a set of TC statements, $Q$ a conjunctive query, and $Q'$ a complete generalization of $Q$ wrt $\C$.
  Then there exists a subquery $Q_0$ of $Q$ such that $Q_0\qcont Q'$ and $Q_0$ is a complete generalization of $Q$.
\end{propositionx}

For our proof we rely on the properties of minimal conjunctive queries.
We remind the reader that $Q$ is \emph{minimal\/} if all its subqueries are strictly more general that $Q$,
or, in other words, 
if $Q$ does not have redundant atoms.
This means that for any subquery $Q_0$ obtained from $Q$ by dropping some of the atoms in $Q$
it holds that $Q_0 \not \sqsubseteq Q$.
Every conjunctive query is equivalent to a minimal query~\cite{Chandra-CQ-containment-77}.

If $\query{Q(\tpl u)}{B}$ is a query, and $\alpha$ a substitution,
then with $\alpha Q$ we denote the query $\query{\tilde Q(\alpha\tpl u)}{\alpha B}$.
We call $\alpha Q$ an \emph{instantiation\/} of $Q$.

Next we identify properties of minimal queries that we use later in the proofs.
In particular, we show that a minimal query $Q$ is complete iff
the image of the frozen query $D_Q$ under $T_\C$ comprises all of $D_Q$.
Moreover we show that any instantiation of a complete minimal query is again complete.

\begin{lemmax}[Completeness of Minimal Queries]
\label{lemma-completeness:of:minimal:queries}
Let $\C$ be a set of TC statements,
$\query{Q(\tpl u)}{B}$ a minimal conjunctive query, 
and $\alpha$ a substitution. Then
\begin{enumerate}
\item $\C \models \qc Q \quad\IFF\quad \cdb Q \incl T_\C(\cdb Q);$
      \label{Claim-All:Atoms:Mapped}
\item $\C \models \qc Q \quad \Longrightarrow \quad \C \models \qc {\alpha Q}$.
      \label{Claim-Instantiations:Are:Complete}
\end{enumerate}
\end{lemmax}

\begin{proof}
Claim 1.
$(\Leftarrow)$ This direction holds for all conjunctive queries:
if $D_Q \subseteq T_\C(D_Q)$, then the freezing mapping $\theta$ trivially satisfies
$\theta B\subseteq D_Q$, hence, $\theta$ is an assignment that satisfies $B$ over $D_Q$
and returns $\theta \bar u$ as an answer.
Thus, $\theta\bar u\in Q(D_Q) \subseteq Q(T_\C(D_Q))$, so that $Q$ is complete wrt $\C$
according to Theorem~\ref{theo-plain:reasoning}.

$(\Rightarrow)$
Let $\tilde B$ be the set of atoms obtained by unfreezing $T_\C(D_Q)$, that is, 
$\tilde B = \theta^{-1}T_\C(D_Q)$,
and define $\tilde Q$ by the rule $\query{\tilde Q(\tpl u)}{\tilde B}$.
Note that $\tilde Q$ is a subquery of $Q$.
Since $Q$ is complete, Theorem~\ref{theo-plain:reasoning} implies that $\theta \bar u \in Q(T_\C(D_Q))$.
It follows that all variables in $\tpl u$ occur in $\tilde B$ so that
the query $\query{\tilde Q(\tpl u)}{\tilde B}$ is safe.
It also follows that $\theta \bar u \in Q(\theta \tilde B) = Q(D_{\tilde Q})$,
which by Proposition~\ref{prop-containment} entails that $\tilde Q\qcont Q$.
However, as $Q$ is minimal, this is only possible if $B = \tilde B$, so that
$D_Q = \theta B = \theta\tilde B = T_\C(D_Q)$.

Claim 2.
Since the \quotes{$\Leftarrow$} direction of Claim~1 holds for all CQs,
it suffices to show that $D_{\alpha Q}\subseteq T_\C(D_{\alpha Q})$, which amounts to showing
$\theta\alpha B \subseteq T_\C(\theta \alpha B)$.
To this end, let $R(\tpl t)$ be an atom of $B$.
Then $R(\theta \tpl t) = \theta R(\tpl t) \in T_\C(\theta B)$ by Claim 1.
This is only possible if there exists a TC statement $C = \tc{R(\tpl s)}G$ that is applicable to $\theta R(\tpl t)$.
Applicability implies there is an assignment $\beta$ such that
$R(\beta \tpl s)  = R(\theta \tpl t)$ and $\beta G \subseteq \theta B$.

We can use $C$ also to map $\theta\alpha R(\tpl t) = R(\theta\alpha\tpl t)$ from $\theta\alpha B$,
using instead of $\beta$ the assignment $\theta\alpha\theta^{-1}\beta$.
This new assignment essentially instantiates $\beta$ by $\alpha$,
but also unfreezes the frozen variables introduced by $\beta$
and finally freezes again all variables.

With this assignment we have
$R(\theta\alpha\theta^{-1}\beta\tpl s) = R(\theta\alpha\theta^{-1}\theta\tpl t)
= R(\theta\alpha\tpl t) \in \theta\alpha B$.
Moreover, since $\beta G \subseteq \theta B$, we also have
$\theta\alpha\theta^{-1}\beta G \subseteq \theta\alpha\theta^{-1}\theta B = \theta\alpha B$.
Consequentely, $C$ is applicable to $\theta\alpha R(\tpl t)$ in $\theta\alpha B$ and
$\theta\alpha R(\tpl t) \in \theta\alpha B = D_{\alpha Q}$.
\end{proof}

Note that Claim~2  does not hold for non-minimal queries.
For example, for the non-minimal query $Q(X) \la R(X,a), R(X,Y)$,
the TCS $C=\Compl(R(X,a);\true)$, and the substitution $\alpha = \set{Y\mapsto c}$
we have that $\set{C}\models \Compl(Q)$, but  $\set{C} \not \models \Compl(\alpha Q)$.

Now, we are ready to prove that it suffices to concentrate on subqueries when looking for MCGs.

\begin{proof}[Proof of Proposition~\ref{prop-MCGs:Are:Subqueries}]
Consider a conjunctive query $\query{Q(\tpl u)}B$ and a complete generalisation $\query{Q'(\tpl u')}{B'}$ of $Q$.
Since every conjunctive query is equivalent to a minimal one and
a query equivalent to a complete one is also complete, 
we assume without loss of generality that $Q'$ is minimal.

Since $Q\qcont Q'$, there exists a query homomorphism $\delta$ from $Q'$ to $Q$.
Let $B_0 = \delta B$ and define $Q_0$ as $\query{Q_0(\tpl X)}{B_0}$.
Then $Q_0 = \delta Q'$ and
using Claim~\ref{Claim-Instantiations:Are:Complete} of
Lemma~\ref{lemma-completeness:of:minimal:queries} we conclude that $Q_0$ is complete.

Since $Q_0$ is a subquery of $Q$, we have $Q\qcont Q_0$.
Moreover, since $Q'$ is minimal and $\delta$ is a substitution with  $Q_0 = \delta Q'$, we have $Q_0\qcont Q'$.
\end{proof}

While Proposition~\ref{prop-MCGs:Are:Subqueries} leaves us with all subqueries of $Q$ as candidates for MCGs,
we can still do better.
We will modify the monotonic $T_C$ operator on database instances
to a monotonic operator $\GenC$ on queries.
This will allow us to characterize completeness in terms of fixed points.
We will then conclude that MCGs are least fixed points,
are unique up to equivalence, and can be computed by fixed point iteration.

First we define the \emph{generalization operator\/} $\GenC$.
To make the mathematics in this section work,
we consider \emph{generalized conjunctive queries\/} 
$Q(\tpl u) \la B$, where we give up the safety condition that
every head variable (that is, every variable occurring in~$\tpl u$)
also occurs in $B$. 
As for conjunctive queries, the set of answers of a generalized 
conjunctive query $Q$ over an instance $D$ is defined as
$Q(D) = \set{ \alpha\tpl u \mid \alpha B \subseteq D }$, 
where $\alpha$ ranges over all assignments of domain values to variables of $Q$.
Note that if $Q$ is unsafe, that is, 
$Q$ has a head variable not occurring in $B$,
$D$ is an instance, 
and $\alpha$ satisfies $Q$ over $D$,
then for every domain value $c$ the modified assignment $\alpha[x/c]$ 
also satisfies $Q$ over $D$.
Consequently, an unsafe query can have an infinite answer set.
Clearly, unsafe queries are not of practical interest, 
but are needed for the mathematical development.

For a set of TC statements $\C$ 
we define the operator $\GenC$ as follows.
If $Q(\bar u) \la B$ is a generalized conjunctive query, 
then the query $\GenC(Q)$ is computed in four steps: 
\begin{enumerate}
 \item  freeze $B$ to obtain the database instance $D_Q$, 
 \item  compute $D_Q':= T_\C(D_Q)$,
 \item  unfreeze $D_Q'$ to obtain a set of atoms $B'\incl B$,
 \item  finally, set $\GenC(Q) := Q'$ where $\query{Q'(\tpl u)}{B'}$.
\end{enumerate}
Since $D_Q' = T_\C(D_Q) \incl D_Q$, it follows that $B'\incl B$.
Thus, $\GenC(Q)$ is a subquery of $Q$.
Intuitively, $\GenC$ keeps only those atoms of a query body
that are \quotes{complete} wrt $\C$.
Note that even if $Q$ is safe, $\GenC(Q)$ may be unsafe.

Now we point out properties of the $\GenC$ operator that will help us
answer the questions about MCGs posed above.
We first observe that the completeness characterization in
Theorem~\ref{theo-plain:reasoning} can be reformulated 
in terms of fixed points (modulo query equivalence) of~$\GenC$.

\begin{propositionx} 
   \label{prop-characterization:of:gc}
Let $\C$ be a set of TC statements and $Q$, $Q'$ 
be generalized conjunctive queries. Then,
\begin{enumerate}
% \item $Q\qcont \GenC(Q)$;
%   \label{item:subquery}
\item $Q\qcont Q'  \quad\Longrightarrow\quad  \GenC(Q)\qcont\GenC(Q')$;
  \label{item:monotonicity}
\item $\C\models\compl Q \quad\Longleftrightarrow\quad  Q\equiv\GenC(Q)$. 
  \label{item:fixed-point}
\end{enumerate}
\end{propositionx}

\begin{proof}[Proof Idea]
% Claim~\ref{item:subquery} holds, since $\GenC(Q)$ is a subquery of $Q$.
Claim~\ref{item:monotonicity} can be shown using the fact that
the containment \quotes{$Q\qcont Q'$}
implies the existence of a homomorphism $\delta$ from $Q'$ to $Q$
such that $\delta Q'\subseteq Q$
\cite{Chandra-CQ-containment-77}.
For each atom $A \in Q'$,
if $A \in T_\C(D_{Q'})$ then $\theta \delta A \in T_\C(D_Q)$.
Thus, $\delta$ is also a homomorphism from $\GenC(Q')$ to $\GenC(Q)$.

To see Claim~\ref{item:fixed-point},
note that the characterizing condition for completeness in
Theorem~\ref{theo-plain:reasoning},
\quotes{$\theta\tpl u \in Q(T_\C(\cdb Q))$,}
is equivalent to the existence of a query homomorphism from $Q$ to $\GenC(Q)$
and thus to the containment $\GenC(Q)\qcont Q$.
Since, according to Claim~\ref{item:monotonicity}, the converse containment 
$Q\qcont \GenC(Q)$
holds anyway, this yields Claim~\ref{item:fixed-point}.
\end{proof}

We are now in a position to prove that there is at most one MCG.
As the proof solely applies simple principles of order theory,
we first fix a suitable vocabulary.
We remind the reader that a \emph{preorder\/} is a reflexive and transitive binary relation $\preceq$
on some set $S$.
A \emph{least element\/} of $S$ is an $l \in S$ such that $l\preceq s$ for all $s\in S$ (note that $S$ may admit several least elements).
The \emph{equivalence relation} $\approx$ \emph{induced by} $\preceq$ is defined over $S$ by $s\approx s'$ iff $s\preceq s'$ and $s'\preceq s$.
A function $f\col S\to S$ is \emph{monotonic\/} if $f(s)\preceq f(s')$ whenever $s\preceq s'$.
An element $s\in S$ is a \emph{fixed point of} $f$ \emph{modulo} $\approx$ if $f(s) \approx s$.
% Clearly, if $s\approx s'$ and $s$ is a fixed point of $f$, then $s'$ is also a fixed point of $f$.
With $f^0$ we denote the identity function on $S$ while $f^i$, where $i\geq 0$,
denotes the $(i-1)$-fold composition of $f$ with itself.
The following lemma is folklore, but is usually expressed for orders (i.e. antisymmetric preorders),
so we write it explicitly for the sake of self-containment:
\begin{lemmax}
  \label{lem-least:fixpoints}
  Let $S$ be a finite set with preorder $\preceq$ and a least element $\bot$, and
  let $f\col S\to S$ be a monotonic function.
  Then
  \begin{enumerate}
  \item $f$ has a least fixed point $p_0$ modulo $\approx$, and
  \item $p_0 = f^{k}(\bot)$, where $k = \min\set{ i \mid f^{i+1}(\bot) \preceq f^i(\bot)}$.
  \end{enumerate}
\end{lemmax}

\begin{proof}
Since $\bot$ is a least element, $\bot\preceq f(\bot)$.
Because $f$ is monotonic, this implies $f(\bot) \preceq f(f(\bot))$.
% i.e. 
% f^1 \preceq f^2(\bot)$.
Inductively, we conclude that $f^i(\bot) \preceq f^{i+1}(\bot)$ for all~$i\geq 0$.
Next, because $S$ is finite, 
% at some point in time 
the sequence $f^i(\bot)$ must enter a loop.
So there exists a least $k\geq 0$ such that 
$f^{k+1}(\bot) \preceq f^{k}(\bot)$.
Together with the fact that 
$f^{k}(\bot) \preceq f^{k+1}(\bot)$,
this implies that
$f^{k}(\bot) \approx f^{k+1}(\bot)$,
therefore 
$p_0 = f^k(\bot)$ is a fixed point modulo $\approx$.

Let $p$ be another fixed point modulo $\approx$. Then $\bot \preceq p$,
and consequently (since $f$ is monotonic), $p_0 = f^k(\bot) \preceq f^k(p) \approx p$.
So $p_0$ is a least fixed point modulo $\approx$.
\end{proof}

We can use Lemma~\ref{lem-least:fixpoints} to rephrase Proposition~\ref{prop-characterization:of:gc},
using the operator $\GenC$ as the function $f$,
and leveraging the fact that containment is a preorder on CQs.
% by stating that 
% (1) $\GenC$ is monotonic wrt containment,
% and
% (2) a conjunctive query is complete wrt $\C$ iff it is a fixed point of $\GenC$.
Precisely, for a given $Q$, we denote the set of subqueries of $Q$,
including those that are not safe, as $\SUB Q$.
Clearly, $\SUB Q$ is finite,
containment is a preorder on $\SUB Q$ (whose induced equivalence relation is query equivalence),
$Q$ is a least element of $\SUB Q$, and
$\GenC$ is a monotonic function on $\SUB Q$.
Thus, all prerequisites of Lemma~\ref{lem-least:fixpoints} are satisfied
and we conclude the following proposition.

\begin{propositionx}
  \label{prop-generalization1}
Let $Q$ be a conjunctive query and $\C$ be a set of TC statements. Then

\vspace{-1ex}

\begin{enumerate}[label=(\alph*)]
\item\label{Claim-fixed:point} 
$\GenC$ has a least fixed point (modulo $\equiv$) $\tilde Q$ in $\SUB Q$ ;
        
\item\label{Claim-completeness} 
$\tilde Q$ is complete;
        
\item\label{Claim-linear:number:of:steps} 
$\tilde Q = {\GenC}^n(Q)$ where $n$ is less or equal than the number of atoms in $Q$;
        
\item \label{Claim-minimality} 
if $Q'$ is any complete generalization of $Q$, then $\tilde Q \qcont Q'$;
       
\item \label{Claim-MCG:if:safe} 
if $\tilde Q$ is safe, then $\tilde Q$ is the MCG of $Q$ wrt $\C$,\ \\
      otherwise no complete generalisation of $Q$ wrt $\C$ exists.
     
\end{enumerate}
\end{propositionx}

\begin{proof}
% Claim~\ref{Claim-fixed:point} hold by the definition of $\tilde Q$,
Claim~\ref{Claim-fixed:point} immediately follows from Lemma~\ref{lem-least:fixpoints},
while Claim~\ref{Claim-completeness} holds because fixed points of $\GenC$ are complete (by Proposition~\ref{prop-characterization:of:gc}).
Claim~\ref{Claim-linear:number:of:steps} holds because the application of $\GenC$
to an incomplete query causes the removal of at least one atom.
Claim~\ref{Claim-minimality} holds, since by Proposition~\ref{prop-MCGs:Are:Subqueries},
every complete generalization of $Q$ has some complete query $Q_0\in\SUB Q$ as a subquery,
and by Proposition~\ref{prop-characterization:of:gc},
$\tilde Q$ is contained in every complete query $Q_0\in\SUB Q$.
Finally, if $\tilde Q$ is not safe, then there is no safe subquery of $Q$ that is complete
and therefore no complete generalisation of $Q$.
Hence, Claim~\ref{Claim-MCG:if:safe} holds.
\end{proof}

A possible algorithm suggested by Proposition~\ref{prop-generalization1} would repeatedly apply the $\GenC$ operator to $Q$,
producing a sequence of subqueries $Q_i = {\GenC}^i(Q)$,
and stop when $Q_{i+1} \qcont Q_i$.
Alternatively, the termination condition could be a check that $Q_i$ is complete.
A third possibility (which we use below in Algorithm~\ref{alg:generalization})
is
% as we show below, 
to stop when $Q_{i+1} = Q_i$, that is,
when $\GenC$ does not remove any atom from $Q_i$.
The following result shows that this is correct: 

\begin{propositionx}
  \label{prop-generalization2}
Let $Q$ be a CQ, $\C$ be a set of TC statements, and $Q_i = {\GenC}^i(Q)$ for $i\geq 0$.
If $Q_{i+1}  = Q_i$, then $Q_i$ is a least fixed point of $\GenC$ modulo $\equiv$ in $\SUB Q$.
\end{propositionx}

\begin{proof}
By Proposition~\ref{prop-generalization1}, repeatedly applying $\GenC$ to $Q$ leads to a first $Q_k$
that is a fixed point modulo $\equiv$,
this query $Q_k$ is a least fixed point modulo $\equiv$,
and $Q_j \equiv Q_k$ for each $j\geq k$.
% the query $Q_j$ is equivalent to $Q_k$.
If $Q_i = Q_{i+1}$, then $i\geq k$,
so $Q_i \equiv Q_k$, therefore $Q_i$ is also
a least fixed point modulo $\equiv$.
\end{proof}

\paragraph{Generalization Algorithm}
Following Proposition~\ref{prop-generalization1}, we construct Algorithm~\ref{alg:generalization}, which computes an MCG for a query $Q$ and TCSs $C$.

\begin{algorithm}%[H]
 \SetKwInOut{Input}{Input}\SetKwInOut{Output}{Output}
 \Input{\ query $\query{Q(\tpl u)}B$, set of TCSs $\C$}
 \Output{\ MCG of $Q$ if one exists; otherwise~$\mathtt{null}$}
 $Q_{\mathit{old}} := Q$, $Q_{\mathit{new}} = \GenC(Q)$;\\
 \While{$Q_{\mathit{new}}$ is safe and $Q_{\mathit{new}} \neq Q_{\mathit{old}}$}{
  $Q_{\mathit{old}} := Q_{\mathit{new}}$\\
  $Q_{\mathit{new}} := \GenC(Q_{\mathit{old}})$\\
 }
 \eIf{$Q_{\mathit{new}}$ is not safe}{
   \KwRet{$\mathtt{null}$}
  }
  {
   \KwRet{$Q_{\mathit{new}}$}
  }
 \caption{Computes an MCG if one exists}
 \label{alg:generalization}
\end{algorithm}

% \paragraph{Algorithm running time}
% The algorithm runs in 
% $\O(|Q|(|Q|^{|\C|}+|Q|))= \O(|Q|^{|\C|+1})$ time for a query $Q$ and a set of TCSs $\C$.
% This is because the algorithm repeats at most $|Q|$ times the following: 
% check if the current query $Q$ is complete which means
% evaluate each TCS in $\C$ against $D_Q$ (the frozen $Q$) which is in $\O(|Q|^{|\C|})$
% and then check if the resulting new $Q$ has fewer atoms than the current $Q$,
% which is in $\O(|Q|)$.

\paragraph*{Computational Complexity}
  \label{gen:par:complexity}
% We define the \emph{MCG\/} problem as the following  decision problem:
% \emph{Given queries $Q$ and $Q'$, and
% a set of TC statements $\C$,
% check if $Q'$ is an MCG of $Q$ wrt $\C$.}
We study two decision problems related to generalizing a (possibly) incomplete CQ.
The first problem estimates the cost of the $\GenC$ operator.
We show that this problem is complete for the complexity class $\DP$,
which intuitively consists of all problems
that can be decided by an algorithm that %simultaneously 
performs two calls to an NP-oracle:
the algorithm responds \quotes{yes} if the first oracle call returns \quotes{yes} and 
the second one returns \quotes{no}.

\begin{propositionx}
\label{prop:gen-dp-complete}
For a set $\C$ of TC statements and two CQs $Q$ and $Q'$,
% and a subquery $Q'$ of $Q$, 
deciding whether $Q' = \GenC(Q)$ is $\DP$-complete.
\end{propositionx}

Our second decision problem estimates the cost of computing an MCG.
We show that this problem is in $\textrm P^{\NP}$,
which is the class of problems
that can be decided in polynomial time 
by a Turing machine that uses an $\NP$-oracle.%
\footnote{An example of a complete problem problem for $\textrm P^{\NP}$,
provided by Krentel (1988), is the lexicographically last satisfying assignment of a Boolean formula:
given a Boolean formula $\phi(X_1,\dots,X_n)$, decide whether variable $X_n$ takes value $1$ in the lexicographically largest satisfying assignment for $\phi$.}
We do not know whether the problem is also $\textrm P^{\NP}$-hard.

\begin{propositionx}
For a set $\C$ of TC statements and two CQs $Q$ and $Q'$,
deciding whether $Q'$ is an MCG of $Q$ wrt $\C$ is in $\textrm P^{\NP}$.
\end{propositionx}
\begin{proof}[Proof sketch]
In order to decide whether $Q'$ is an MCG of $Q$,
it is sufficient to execute Algorithm~\ref{alg:generalization},
verify whether the output query $Q_{\mathit{new}}$ is \texttt{null}, and if it is not,
determine whether $Q_{\mathit{new}} \equiv Q'$.
The latter can be done with a call to an $\NP$ oracle,
since equivalence of CQs is known to be in $\NP$ (and the size of $Q_{\mathit{new}}$ is bounded by the size of $Q$).

Next, we observe that Algorithm~\ref{alg:generalization} performs a number of iterations
that is linear in the size of $Q$ (because the $\GenC$ operator intuitively discards atoms from the body of its input query).
So to complete the proof, it is sufficient to observe that each of these iteration can be executed in linear time,
assuming an $\NP$ oracle. 
More precisely, for any query $Q_\mathit{old}$,
the query $\GenC(Q_\mathit{old})$ can be computed by retaining certain atoms in the body $B$ of $Q_\mathit{old}$:
retain an atom $A$ iff there is a TCS $\Compl(A;G)$ in $\C$ and a substitution $\alpha$ such that $\alpha G \subseteq B$.
Deciding the existence of such a substitution is in $\NP$, because \ei the size of $\alpha$ is bounded by the size of $G$ (so it can be encoded as a certificate with polynomial length),
and \eii $\alpha G \subseteq B$ can be trivially verified in polynomial time.

% with linearly many calls to an $\NP$ oracle.
% Indeed,

% using a $\DP$-oracle.
% using an $\NP$-oracle.

% A call to a $\DP$-oracle can be realized by two calls to an $\NP$-oracle. 
% Thus, the MCG decision problem belongs to the class of problems 
% $\textrm P^{\NP}$.
 \end{proof}

\section{Query Specialization}
\label{sec:query:specialization}
As discussed in the introduction, another common way to approximate a query is to make it more specific.
% A query $Q'$ \emph{specializes} a query $Q$ if it is contained in $Q$ (i.e. if $Q' \sqsubseteq Q$).
% And a specialization $Q'_1$ is preferred to a specialization $Q'_2$ if it is more general, i.e. if $Q'_2 \sqsubseteq Q'_1$.
%
Analogously to the previous section,
we study the existence and computation of CQs that specialize $Q$ and are complete wrt a set of TCSs.
As in the preceding section, all queries mentioned are assumed to be conjunctive queries.

We say,
% \begin{definitionx}[Complete Specialization]
that $Q'$ is a \emph{complete specialization\/} (CS)
of $Q$ wrt a set $\C$ of TCSs if $Q' \sqsubseteq Q$ and $\C \models \compl {Q'}$.
% \end{definitionx}
Among those queries, the preferred ones are the \emph{maximal\/} ones:
\begin{definitionx}[Maximal Complete Specialization]
A query $Q'$ is a \emph{maximal complete specialization\/} (MCS) 
of $Q$ wrt a set $\C$ of TCSs if:  
\begin{itemize}
  \item $Q'$ is a CS of $Q$ wrt $\C$,
   and
  \item it is a maximal one, that is, there exists no query $Q''$ such that 
     $Q''$ is a CS of $Q$ and
     $Q'$ is strictly contained in $Q''$.
	% \item no CS $Q''$ of $Q$ wrt $\C$ satisfies $Q' \sqsubset Q''$.

\end{itemize}
\end{definitionx}

% In this section, we define formally the notion of query specializations. 
% We investigate properties of query specializations and 
% we identify algorithms that compute them.
%
%
% We discuss how many query specializations a query can have 
% We establish the 
%
% Three parts:
%
% * Definition of the problem; Motivations Examples; etc.
% * Properties
% * Complexity theorems: `MCS(Q,C;Q')`? 
% 	* in the size of `C`
% 	* in the size of `Q, C` and `Q'`
% * Algorithms:
% 	* an algorithm that computes a single complete specialization
% 	* an algorithm that computes a single minimal complete specialization
% 	* an algorithm that computes all minimal complete specialization
%
% A query $Q'$ is a \emph{complete specialization} of a query $Q'$ 
% wrt a set of TCSs $\C$ if $Q' \qcont Q$ and $\C \models \compl {Q'}$.

Even though the definition of an MCS is symmetric to the one of an MCG,
it will turn out that there are important differences between these two settings
that require different approaches.
First, a query may have several non-equivalent 
MCSs (whereas it admits at most one MCG up to equivalence).
Second, a query may admit CSs but no MCS:
\begin{theoremx}\label{th:no_mcs}
There is a query $Q$ and set $\C$ of TCSs such that $Q$ admits CSs wrt $\C$, none of which is maximal.
\end{theoremx}
\begin{proof}%[Proof(sketch)]

Consider a schema with a binary relation $\conn$ that specifies whether two cities are connected via a direct flight,
and let $C$ be the TCS $\Compl(\conn(X,Y); \conn(Y,Z))$,
which states that the data is complete for all (direct) connections that can be extended.
The query $\query{Q(X)}{\conn(X,Y)}$, which retrieves all cities with an outgoing flight, is not complete wrt $\{C\}$.
As an illustration,
consider the incomplete database 
$\IDB = (\di, \da)$ where $\di = \{ \conn(a,b), \conn(b,c), \conn(d,e) \}$ and $\da = \{ \conn(a,b), \conn(b,c) \}$.
The database $\IDB$ satisfies $C$, but $Q(D^i) \not \subseteq Q(D^a)$.

Now let $Q'$ be any complete specialization of $Q$ wrt $\C$.
Because $Q'$ is complete, its body must contain a set $\A_k$ of atoms that identifies a round trip of length $k$ for some $k \ge 1$,
%defined (up to variable renaming) as $\A_k = \conn(X_0,X_1),\conn(X_1,X_2),\dots,\conn(X_{k-1},X_0)$,
having the form $\A_k = \set{\conn(X_0,X_1),\conn(X_1,X_2),\ldots,\conn(X_{k-1},X_0)}$,
and $Q'$ must project exactly one of these variables.

Let $k'$ be the largest $k$ for which this holds, and consider the query $Q_{2k'}(X_0) \la \A_{2k'}$.
This query is also a complete specialization of $Q$.
Moreover, the function $\delta_k$ that maps each $X_i$ to $X_{i \bmod k}$ is a homomorphism from $Q_{2k'}$ to $Q'$,
which implies $Q' \sqsubseteq Q_{2k'}$.
However, $\A_k$ cannot be mapped to $\A_{2k'}$ (i.e., the body of $Q$),
therefore there is no homomorphism from $Q'$ to $Q_{2k'}$,
which implies $Q_{2k'} \not\sqsubseteq Q'$.
So $Q'$ is not an MCS. 
\end{proof}

% An immediate consequence of Theorem \ref{th:no_mcs} is that the set of CSs of a CQ may be infinite.
%
These observations lead us to investigate restrictions on a set $\C$ of TCSs
that guarantee the set of CSs (thus also of MCSs) of a query wrt $\C$ to be finite.
We define the \emph{dependency graph} of $\C$ as the graph whose nodes are the relation names appearing in $\C$,
and with an edge from $R$ to $R'$
iff $R'$ appears in $G$ for some statement of the form $\Compl(A; G)$, where the atom $A$ is over $R$.
We say that $\C$ is \emph{acyclic} if its dependency graph is acyclic.
The following result provides 
an upper bound on the size of an MCS for such a $\C$.

\begin{theoremx}\label{th:bound_mcs_atoms}
\label{th:max-size-mcs}
Let $Q$ be a CQ, $\C$ an acyclic set of TCSs and $s$ the number of relation names 
that appear in $\C$.
% such that exactly $s$ relations names appear in $\C$.
% and $|Q|$ the size of query.
% Let $|\C|$ be the size of $\C$, %, $r$ the maximal arity in $\C$, 
Then the number of atoms in a MCS of $Q$ wrt $\C$ 
is in $\O(|Q|\times|C|^{s+1})$.%
\footnote{This acyclicity requirement on $\C$ can be relaxed,
while guaranteeing  a bound on the size of an MCS. 
In particular,
it is sufficient to require that $\C$ is \emph{weakly acyclic}, 
as defined in~\cite{fagin_data_exchange}.}
\end{theoremx}

In the rest of this section,
we show how to compute the CSs of $Q$ that are maximal (wrt $\sqsubseteq$) within the space of queries with at most $|Q| + k$ atoms, 
for some $k \in \nn$.
We call such a query a $k$-MCS.
Theorem \ref{th:bound_mcs_atoms} implies that, if $\C$ is acyclic, then MCSs and $k$-MCSs coincide for a large enough $k$.
We also restrict our investigations to the case where $Q$ is a minimal query in the sense of Lemma \ref{lemma-completeness:of:minimal:queries},
that is $Q$ has no redundant atoms.
%(meaning that no CQ obtained by deleting atoms from the body of $Q$ is equivalent to $Q$).

We first analyze how to compute MCSs without adding atoms to the body of $Q$ (we call such MCSs maximal complete instantiations, defined below),
and we then extend this approach to find $k$-MCSs.

%-----------------------------------%
\subsection{Maximal Complete Instantiations}
%-----------------------------------%

% We can now use this notion to devise a procedure that computes $k$-MCSs. 
% %
% At its core,
% this procedure relies on a sub-routine that specializes a query by replacing variables with variables or constants.
% We call these specializations \emph{images}, and compute the maximal complete ones, defined as follows:

\begin{definitionx}[Maximal Complete Instantiation]
% If $Q$ is a CQ, we call \emph{instantiation} of $Q$ any query $\alpha Q$ obtained by applying a substitution $\alpha$ to $Q$.
If $\C$ is a set of TCSs, then a \emph{maximal complete instantiation (MCI)} of $Q$ wrt $\C$ is an instantiation of $Q$ that is complete wrt $\C$
and maximal wrt containent among these.
\end{definitionx}
As an illustration, 
in the proof sketch of Theorem~\ref{th:no_mcs},
the query $Q'(X) \la \conn(X,X)$ is the only MCI for $Q$ wrt $\C$.

Our approach to compute MCIs revolves around a specific type of substitution that we call a \emph{complete unifier},
defined as follows:

\begin{definitionx}[Complete Unifier]\label{def:complete_unifier}
Let $Q$ be a CQ and $\C$ a set of TCSs.
A \emph{complete unifier} for $Q$ and $\C$ is a substitution $\gamma$ such that, 
for each atom $A$ in the body of $Q$,
there is a TCS $C  =\Compl(A';G)\in \C$ that satisfies
% \begin{itemize}
  \[\gamma A = \gamma A' \qquad\te{and}\qquad \gamma G \subseteq \gamma B.\]
% \end{itemize}
\end{definitionx}

Applying a complete unifier to $Q$ yields a complete query:
\begin{propositionx}\label{prop:comp_unifier}
If $\gamma$ is a complete unifier for $Q$ and a set $\C$ of TCSs,
then $\C \models \Compl(\gamma Q)$.
\end{propositionx}

\begin{example-new}
In Ex.~\ref{ex:intro}, the substitution 
$\gamma =\set{L \mapsto \english}$ is a complete unifier 
for $Q_\indexpbl$ and the set of TCSs, 
and $\gamma Q_\indexpbl = Q_\indexpbl^{\textit{spec}}$ is indeed complete.
\end{example-new}

%
%
% We observe that an MCI is a $0$-MCS, but the converse may not hold:
% for instance, 
% the query $Q'(X) \la R(X,X), S(X,W)$ may be an MCI for $Q(X) \la R(X,Y), R(Y,Z)$ (wrt some $\C$) but it cannot be a $0$-MCS.

% A key observation to understand our procedure is that every complete image of $Q$ is more specific than some specialization obtained with a complete unifier:
\begin{theoremx}\label{th:cu_dom}
Let $\C$ be a set of TCSs, let $Q$ be a query, and $Q'$ be a complete instantiation of $Q$ wrt $\C$.
 Then there is a complete unifier $\gamma$ for $Q$ and %some $\C' \subseteq \C$ 
 $\C$
 such that $Q' \sqsubseteq \gamma Q$.
\end{theoremx}
\begin{example-new}[continued]
Consider
$Q'(N) \la \ \pupil(N,1,S), \school(S,\primary,\merano), \learns(N,\english)$. 
Then $Q'$ is a complete instantiation of  $Q_\indexpbl$
and $Q' \sqsubseteq \gamma Q_\indexpbl$.
\end{example-new}

Observe that if $\gamma$ is a complete unifier for $Q$ and $\C$, then there must be a subset $\C'$ of $\C$ such that
$\gamma$ is a complete unifier for $Q$  and 
every atom $A$ in $Q$ is unified with the head of exactly one TCS from $\C'$.
We call such a $\C'$ a \emph{matching subset} of $\C$ for $Q$.
Since $Q$ and $\C'$ admit a complete unifier,  they admit a (unique) most general one.
Let $\mgu(Q,\C')$ denote this unifier,
and let
$\mgu(Q, 2^{\C}) = \{\mgu(Q,\C') \mid \C' \te{is a matching subset of} \C \te{for} Q, Q \te{and} \C' \text{ admit a complete unifier} \}$.
Since $\C$ is finite, $\mgu(Q, 2^{\C})$ is finite, too.
More, as an immediate consequence of Theorem~\ref{th:cu_dom} and Proposition~\ref{prop:comp_unifier},
each MCI of $Q$ wrt $\C$ is equivalent to $\gamma Q$ for some $\gamma \in \mgu(Q, 2^{\C})$.
This is the rationale behind our procedure to compute all MCIs of $Q$ wrt $\C$ described with Algorithm \ref{alg:spec}.
The first loop computes $\{\gamma Q \mid \gamma \in \mgu(Q, 2^{\C})\}$, 
relying on a function $\texttt{MGU}$ that returns $\mgu(Q, \C')$ if it exists,
and \texttt{null} otherwise.
The second loop discards non-maximal instantiations within these.

\begin{algorithm}[!ht]
 \caption{Computes Maximal Complete Instantiations}
 \label{alg:spec}
 \SetAlgoNoEnd
 \LinesNumbered
 \DontPrintSemicolon
 \SetKwInOut{Input}{Input}\SetKwInOut{Output}{Output}
 \Input{a query $Q(\tpl u)\la A_1,\dots,A_n$, a set  $\C$ of TCSs}
 \Output{the set $\S$ of all MCIs of $Q$ wrt $\C$}
 % \Begin{
% \tcc{auxiliary variables}
% let variable $\S$ be a set of CS (will contain $=\{\gamma Q \mid \gamma \in \mgu(Q, 2^{\C})\}$) \\
% let variables $\alpha_1,...,\alpha_n$ be substitutions (one for each query atom)\\
% % let $\alpha_1',...\alpha_n'$ be substitutions (used for TC conditions)\\
% let variable $\alpha$ be a substitution (auxiliary variable) \\
% \tcc{find complete specializations (including non-maximal)}
% $\S\leftarrow \emptyset$ \\
$\S := \emptyset$,  $\gamma := \mathtt{null}$\\
\ForEach{matching subset $\C'$ of $\C$ for $Q$}{
 {
    $\gamma := \mathtt{MGU}(Q,\C')$ \\
    \If{ $\gamma \not= \mathtt{null}$ }{
      $\S := \S \cup \set{\gamma Q}$
    }
  }

}

% \tcc*{eliminate non maximal}
\ForEach{$Q'$ in $\S$}{ %\texttt{eliminate non-maximal specializations}
  \lIf{there exists $Q''\in \S$ such that $Q' \sqsubseteq \set Q''$}
    {$\S := \S \setminus \set{Q'}$}
  % {drop $Q'$ from $\S$}
  % \lElseIf{$Q'' \sqsubseteq Q'$}
  %   {drop $Q$ from $\S$}
}
\KwRet{$\S$}
\end{algorithm}

\paragraph*{Computational Complexity of the MCI Decision Problem}
We study the problem that consists in deciding whether a query is an MCI. 
\begin{theoremx}
Given a query $Q$, a candidate query $Q'$ and a set~$\C$ of TCSs, 
deciding whether $Q'$ is an MCI of $Q$ \wrt $\C$ is in $\PIPTWO$.
% \note{I can also show that it is $\DP$-hard already for Boolean almost propositional queries.}
\end{theoremx}
\begin{proof}[Proof sketch]
We show that a non-deterministic Turing machine can verify in polynomial time that $Q'$ is \emph{not} an MCI, assuming an oracle for $\NP$,
as follows: 
\begin{enumerate}[label=(\Roman*)]
  \item determine whether $Q'$ is complete,~\label{item:mci1}
  \item if it is, then determine whether it is also an instantiation of $Q$,\label{item:mci2}
  \item if it is, then verify that there is a more general complete instantiation of $Q$.~\label{item:mci3}
  \end{enumerate} 

For Step~\ref{item:mci1},
it was show in~\cite{Razniewski:Nutt-Compl:of:Queries-VLDB11} that deciding whether a CQ is complete wrt to a set of TCSs is $\NP$-complete.
Therefore one can determine whether $Q'$ is complete with one call to an $\NP$ oracle.

For Step~\ref{item:mci2}, 
similarly, we show that deciding whether $Q'$ is an instantiation of $Q$ is in $\NP$.
Observe that $Q'$ is an instantiation of $Q$ iff there is a substitution $\alpha$ such that $Q' = \alpha Q$.
If such a substitution exists,
then it is a function from the variables that appear in $Q$ to the variables and constants that appear in $Q'$.
So $\alpha$ can be encoded as a certificate with length polynomial in the size of $Q$ and $Q'$.
And given $\alpha$,
one can (trivially) verify in polynomial time that $Q' = \alpha Q$.

For Step~\ref{item:mci3}, it is sufficient to verify that there exists a substitution $\beta$ such that \ei $\beta Q$ is complete, 
\eii $Q' \sqsubseteq \beta Q$ and
\eiii $\beta Q \not \sqsubseteq Q'$.
Again, such a substitution can be encoded as a certificate with polynomial length.
Next, as we explained above for Step~\ref{item:mci1}, whether \ei holds can be determined (therefore also verified) with one call to an $\NP$ oracle.
And it is well-known that containment of CQs is in $\NP$,
therefore \eii and \eiii can also be verified by calling such an oracle (note that a quantifier alternation is only needed here for \eiii, which is in $\coNP$,
whereas $\ei$ and $\eii$ are in $\NP$).
\end{proof}

It is not hard to show 
(using a similar technique as for Proposition~\ref{prop:gen-dp-complete})
that deciding the MCI problem is $\DP$-hard.
Whether it is $\PIPTWO$-complete is an open question.

%---------------------------------------------------%
\subsection{Adding Atoms}
\label{sec:query-spec-adding-atoms}
% We can now extend our approach to computing all $k$-MCSs of a CQ $Q$ of size $n$ wrt a set $\C$ of TCSs.
% Let $\Sigma_\C$ be the set of all relation names that appear in $\C$.
% %
% We call \emph{extension} of $Q$ any CQ obtained by adding 
% atoms to the body $B$ of $Q$.
% And we call an extension \emph{fresh} if the variables of these added atoms do not appear in $Q$.
% %
% A first observation is that a fresh extension $Q'$ of $Q$ is a specialization of $Q$,
% therefore the MCIs of $Q'$ are CSs of $Q$. 
% %

We can now extend our approach to computing all $k$-MCSs of a query $Q$ of size $n$ wrt a set $\C$ of TCSs.
Let $\Sigma_\C$ be the set of all relation names appearing in $\C$.
An \emph{extension\/} of $Q$ is a query obtained by adding 
atoms to the body of $Q$.
We call an extension \emph{fresh} if the variables of the added atoms do not appear in $Q$.
A first observation is that a fresh extension $Q'$ of $Q$ is a specialization of $Q$,
therefore the MCIs of $Q'$ are CSs of $Q$.

Next, a key observation is that 
the $k$-MCSs of $Q$ coincide with the MCIs of
all fresh extensions of $Q$ with $n+k-1$ fresh atoms,
restricted to MCIs of size $\le n + k$.
The intuition is the following.
Let $Q'$ be a CS of $Q$ of size $\le n+k$.
Then there is a homomorphism $\delta$ from $Q$ to $Q'$.
Let also $B'$ be the set atoms in $Q'$ that have no preimage via $\delta$.
Then the size of $B'$ is $\le n+k-1$.
%
% Based on this observation,
We construct a new query $Q''$ of size $2n+k-1$
whose body consists of the body $B$ of $Q$,
and a fresh version of $B'$.
Then $Q'\qcont Q''$. 
Moreover, since $Q'$ is complete, $Q''$ must have an MCI at least as general as $Q'$.

The first loop of Algorithm~\ref{alg:kspec}  implements this idea:
the function $\texttt{MCI}_{\le n+k}(Q'',\C)$ returns the MCIs of $Q''$ wrt $\C'$ of size $\le n+k$.
Then the second loop discards non-maximal instantiations within these.

%
% So far we have seen that it is possible to compute 
% all MCSs if TCSs statements do not generate infinite chains.
% %
% Conditions on TCSs such as acyclicity guarantee that
% such chains do not exists.
% %
% In this part we show
% how to compute k-MCSs.
% Assume we are given a query $Q$ 
% and we want to compute all $k$-MCSs 
% under set of TC statements $\C$.
% Let $\Sigma_\C$ be a set of all relation names that appear in $\C$.
%
% For convenience, we call $Q'$ an \emph{extension}
% of query $Q$ if $Q'$ is obtained from $Q$ by 
% adding new atoms to the query body.
% The idea is to consider all fresh extensions of 
% $Q$ using relations from $\Sigma_\C$.
% An extension is \emph{fresh} if it made out of fresh atoms, i.e., 
% atoms where all variables are fresh.
%
% We search for k-MCSs by finding  MCIs of fresh extension of sizes 
% $n, n+1, \ldots, k$. 
%
% Each fresh extension is a specialization of $Q$, thus
% their MCIs are going to be also CSs
% for $Q$.
%
% In fact, to find all  k-MCSs, 
% we can show that it is sufficient to observe all extension of 
% size $n+k-1$.
%
% Further, note that fresh extensions contains all the maximal 
% specializations that $Q$ can have  up to size $|Q|+k$.
% %
% Thus, we can show 
% that MCIs of fresh extensions contain all k-MCSs for $Q$.
%
% Then it only remains to eliminate non-maximal ones.
%
% The algorithm that realizes the discussed idea is shown
% in Alg.~\ref{alg:kspec}. 
%To compute MCSs of fresh extensions 
% the algorithm calls $\MCI$ algorithm (Fig.~\ref{alg:spec}).

\begin{algorithm}%[H]
 \label{alg:kspec}
 \SetAlgoNoEnd
 \LinesNumbered
 \DontPrintSemicolon
 \SetKwInOut{Input}{Input}\SetKwInOut{Output}{Output}
 \Input{a query $Q(\tpl u)\la B$, a set $\C$ of TCSs , $k \in \nn_0$}
 \Output{the set $\S$ of all $k$-MCSs of $Q$ wrt $\C$}
 % \Begin{
% let $\S$ be a set of complete specialization of size up to $k$\\
% \tcc{find complete specializations (including non-maximal)}
$\S := \emptyset$ 

% \tcc*[f]{set of complete specialization of size up to $k$} \\
%
% \For(\tcc*[f]{compute specializations up to size $k$})
% {$i\leftarrow 0$ \KwTo $k$}
{\ForEach{set $B'$ of fresh atoms of size $n+k-1$ over
  $\Sigma_\C$ }{
  \textit{construct} $Q''(\tpl u) \la B,B'$\\
  $\S := \S \cup \mathtt{MCI}_{\le n + k}(Q'',\C)$
    % $\S leftarrow \S \cup \MCI(Q'',\C)$
}

}
% \tcc*{eliminate non maximal}
\ForEach{$Q''$ in $\S$}{
% \ForEach(\tcc*[f]{eliminate non-maximal ones}){$Q''$ in $\S$}{
  \lIf{exists $Q'''\in \S$ such that $Q'' \sqsubseteq Q'''$}
  { $\S := \S \setminus \{Q''\}$}
}
\KwRet{$\S$}
% }
 \caption{Computes k-MCSs}
\end{algorithm}

% \input{input/4-specialization-OS}
% !TEX root = ../rr-2024-main.tex

\section{Implementation}
\label{sec:implementation}

We now briefly discuss how we chose platforms to implement our algorithms
in the demo tool MAGIK~\cite{MAGIK-demo-VLDB2013}.

\paragraph{Implementing Generalization}

At its core the generalization algorithm repeatedly applies the $G_\C$ operator to the
database instance $D_Q$ until it has reached a fixed point.
In each round, starting from an instance of the size of the query,
a new instance is produced by applying TC statements in a forward fashion,
until original and new instance are identical.
We implemented this via a datalog engine, namely the ASP solver \texttt{dlv}.
The instance $D_Q$ is represented by facts, initially obtained by freezing the query,
and the TC statements are translated into TC rules.

For example,
the query $Q_\indexppb(N)$ would be translated into the facts $\pupil^i(n',c',s'),$ and $\school^i(s',\primary,\merano)$,
while the statement $C_\indexpb$ would be translated into the rule
$\pupil^a(N,C,S) \leftarrow \pupil^i(N,C,S),\,\school^i(S,T,\merano))$.
To distinguish original and new atoms, the relation symbols are either labeled with the superscript ${}^i$,
standing for ``ideal'', or ${}^a$, standing for ``available''.
A fixed point is reached if each ``ideal'' fact is translated into an ``available'' fact.

While the ASP functionality is not used here, it becomes beneficial when taking account of disjunctive constraints
such as finite domain constraints.

\newcommand{\ignore}[1]{}

\ignore{
The generalization algorithm has been implemented in the demo tool MAGIK~\cite{MAGIK-demo-VLDB2013}.
The implementation (in Java) calls the $G_\C$ operator until its least fixpoint is reached. 
The operator $G_\C$ is implemented via the ASP solver \texttt{dlv}.
% in the same way as we implement the $T_C$ operator. 
We observe that one can further improve the implementation of the algorithm by encoding it completely in Datalog and using recursion. 
Basically, one would need to extend every relation with an extra argument 
$k$ that represents the ``version'' of the atoms obtained after applying ${G_\C}^k$.
Then, the recursion can stop as soon as we encounter the first $k$ for which the generalization query is complete.
}

\paragraph{Implementing Specialization}

The core operation of the specialization algorithm in Section~\ref{sec:query:specialization} is unification,
which makes ASP systems unsuitable as a platform while it is offered as a functionality by Prolog.
We implemented it in SWI-Prolog.
The problem is inherently hard, due to the doubly exponential search space.

For instance, let us consider the query
$Q_\indexl \la \learns(N,L)$, 
and the set of TCSs from our running example,
minus the TCS $C_\indexpb$,
and extended with
$\Compl(\pupil(N,C,S);\class(C,S,L,\halfDay))$ 
\text{~and~} 
$\Compl(\pupil(N,C,S);\class(C,S, L,\fullDay)).$
The search for all $3$-MCSs would run out of memory.%
\footnote{Using an \texttt{Intel Core i7} with 8GB of RAM, and 2GB allocated to SWI-Prolog. }

To avoid this we implemented several optimizations.%
\footnote{The code is available at: \url{https://github.com/osavkovic/QuerySpecProlog/}}
Briefly, to compute $k$-MCSs,
we consider extensions of size 0, then $1,2,\ldots,k$.
% and keep in memory the extensions from smaller to bigger: first 
For $j>i$,
the $i$-MCSs are likely to be identical to the $j$-MCSs.
By keeping in memory the list of maximal specializations collected so far, 
we can compare them for containment with the currently analyzed extensions.
This step reduces the search space at line 2 of Alg.~\ref{alg:kspec}.
This way, we reduce the number of specializations that we store, and this reduces memory consumption.
Further optimizations can be implemented in the presence of integrity constraints (e.g., foreign keys) which we do not present in this paper.

\ignore{
Specialization can be implemented by the algorithm in Section~\ref{sec:query:specialization}.
This algorithm has also been implemented in MAGIK~\cite{MAGIK-demo-VLDB2013}, this time via SWI-Prolog.
The problem is inherently hard, due to the doubly exponential search space.

For instance, let us consider the query
$Q_\indexl \la \learns(N,L)$, 
and the set of TCSs from our running example,
minus the TCS $C_\indexpb$,
and extended with
$\Compl(\pupil(N,C,S);\class(C,S,L,\halfDay))$ 
\text{~and~} 
$\Compl(\pupil(N,C,S);\class(C,S, L,\fullDay)).$
The search for all $3$-MCSs would run out of memory.%
\footnote{Using an \texttt{Intel Core i7} with 8GB of RAM, and 2GB allocated to SWI-Prolog. }

To avoid this we implemented several optimizations.%
\footnote{The code is available at: \url{https://github.com/osavkovic/QuerySpecProlog/}}
Briefly, to compute $k$-MCSs,
we consider extensions of size 0, then $1,2,\ldots,k$.
% and keep in memory the extensions from smaller to bigger: first 
For $j>i$,
the $i$-MCSs are likely to be identical to the $j$-MCSs.
By keeping in memory the list of maximal specializations collected so far, 
we can compare them for containment with the currently analyzed extensions.
This step reduces the search space at line 2 of Alg.~\ref{alg:kspec}.
This way, we reduce the number of specializations that we store, and this reduces memory consumption.
Further optimizations can be implemented in the presence of integrity constraints (e.g., foreign keys) which we do not present in this paper.} 

We performed preliminary tests of our optimized code, and reached size $|Q_\indexl|+7$ after around more than 2 hours.
% since it went out of global stack that was set to 2GB. 
%
The average running times (over multiple runs) are reported in Table~\ref{table:specialization-test}.
We observe that the running time grows exponentially with the number of atom added to the query.

\begin{table}[!ht]
\begin{center}
\begin{tabular}{l@{\qquad}c@{~~~~}c@{~~~~}c@{~~~~}c@{~~~~}c@{~~~~}c@{~~~~}c@{~~~~}c@{}}
\toprule
\textbf{k-MCS}	& 
0	&
1	&
2	&
3	&
4	&
5	&
6	&
7	
\\%[-0.5em] %\midrule
\textbf{CPU time (sec)}	& 
0			&
0			&
0		&
0		&
0 		&
8		&
725		&
9083%$^\dagger$	
\\
\bottomrule
\end{tabular}
\end{center}
\caption{Time required for the specialization algorithm to compute 
k-MCS of query $Q_\indexl$. 
%Symbol 9082.942$^\dagger$ means that machine runs out of global stack (of 2GB) after this time.
}
\vspace*{-2em}
\label{table:specialization-test}
\end{table}%
This small experiment shows 
% two things. 
that the optimizations improve the initially proposed specialization algorithm (especially,  concerning memory),
and that such optimizations may still not be sufficient for large queries.

% !TEX root = ../rr-2024-main.tex

\section{Conclusion}
\label{sec:conclusion}

In this work, we studied the completeness of conjunctive queries over partially complete databases
where completeness is determined via so-called completeness statements.
For queries that cannot be answered completely, we study ways to approximate such queries with more general
or more special queries that are complete.
In particular, we established characterizations and algorithms for finding maximal complete specializations (MCSs)
or the (unique) minimal complete generalization (MCG).
The MCG can be characterized as the least fixed-point of a monotonic operator in a preorder.
An MCS can be computed through unification between the query and completeness statements.
The complexity of both problems is studied, and implementation techniques using ASP and Prolog engines are discussed.

We plan to extend our theory by considering integrity constraints like primary and foreign keys,
and finite domain constraints, and to enhance our implementation techniques.

%----------------------------------------------%
% Bibliography
%----------------------------------------------%

\bibliographystyle{plain} % the recommended bibstyle
{\footnotesize
\bibliography{bib/ognjen-bib,bib/ognjen-new,bib/os-bib,bib/ref-stability}

\begin{thebibliography}{10}

\bibitem{Arenas:Aggregates:IncosistentDBs:Journal:2003}
Marcelo Arenas, Leopoldo Bertossi, Jan Chomicki, Xin He, Vijay Raghavan, and Jeremy Spinrad.
\newblock Scalar aggregation in inconsistent databases.
\newblock {\em Theor. Comput. Sci.}, 296(3), 2003.

\bibitem{Chandra-CQ-containment-77}
Ashok~K. Chandra and Philip~M. Merlin.
\newblock Optimal implementation of conjunctive queries in relational data bases.
\newblock STOC '77, page 77–90. Association for Computing Machinery, 1977.

\bibitem{fagin_data_exchange}
Ronald Fagin, Phokion~G. Kolaitis, Ren{\'{e}}e~J. Miller, and Lucian Popa.
\newblock Data exchange: Semantics and query answering.
\newblock In {\em Proc. ICDT}, pages 207--224, 2002.

\bibitem{Fan:Geerts-relative_information_completeness:pods:09}
W.~Fan and F.~Geerts.
\newblock Relative information completeness.
\newblock In {\em PODS}, 2009.

\bibitem{Fan:Geerts-capturing_missing_tuples_and_values:pods:10}
W.~Fan and F.~Geerts.
\newblock Capturing missing tuples and missing values.
\newblock In {\em PODS}, 2010.

\bibitem{fensel2020knowledge}
Dieter Fensel, Umutcan Simsek, Kevin Angele, Elwin Huaman, Elias K{\"a}rle, Oleksandra Panasiuk, Ioan Toma, J{\"u}rgen Umbrich, and Alexander Wahler.
\newblock {\em Knowledge graphs: Methodology, Tools and Selected Use Cases}.
\newblock Springer, 2020.

\bibitem{Libkin:SQL:Null:PODS2016}
Paolo Guagliardo and Leonid Libkin.
\newblock Making sql queries correct on incomplete databases: A feasibility study.
\newblock PODS '16, 2016.

\bibitem{levy_completeness}
A.Y. Levy.
\newblock Obtaining complete answers from incomplete databases.
\newblock In {\em Proc. VLDB}, pages 402--412, 1996.

\bibitem{motro_integrity}
A.~Motro.
\newblock Integrity = {V}alidity + {C}ompleteness.
\newblock {\em ACM TODS}, 14(4):480--502, 1989.

\bibitem{CIKM2015-Nutt}
Werner Nutt, Sergey Paramonov, and Ognjen Savkovic.
\newblock Implementing query completeness reasoning.
\newblock In {\em {CIKM}}, pages 733--742, 2015.

\bibitem{papadimitriou-complexitybook}
Christos~M. Papadimitriou.
\newblock {\em {Computational complexity}}.
\newblock Addison-Wesley, 1994.

\bibitem{Razniewski-2024-Comple-Survey}
Simon Razniewski, Hiba Arnaout, Shrestha Ghosh, and Fabian Suchanek.
\newblock Completeness, recall, and negation in open-world knowledge bases: A survey.
\newblock {\em ACM Comput. Surv.}, 56(6), 2024.

\bibitem{Razniewski:Nutt-Compl:of:Queries-VLDB11}
Simon Razniewski and Werner Nutt.
\newblock Completeness of queries over incomplete databases.
\newblock {\em {PVLDB}}, 4(11):749--760, 2011.

\bibitem{MAGIK-demo-VLDB2013}
Ognjen Savkovi\'c, Paramita Mirza, Alex Tomasi, and Werner Nutt.
\newblock Complete approximations of incomplete queries.
\newblock In {\em VLDB Endowment 6 (12) (VLDB) (2013)}, pages 1378--1381, 2013.

\end{thebibliography}
}

%----------------------------------------------%
% Supplemantary Material
% Comment/Uncomment to include/exclude it 
%----------------------------------------------%

\newpage
\appendix
% \appendixpage
% \section*{Supplementary Material for the Paper:  \\
% Complete Approximations of Incomplete Queries}

\section*{Appendix}

\startcontents[chapter]
\printcontents[chapter]{l}{0}{\setcounter{tocdepth}{1}}

\section{The Complexity of Identifying MCGs}
Recall that $\DP$ consists of problems
that can be decided by an algorithm that simultaneously performs two calls to 
an NP-oracle. The algorithm responds \quotes{yes} if the first oracle call returns \quotes{yes} and 
the second one returns \quotes{no}.
Our proof is based on a reduction of a $\DP$-complete graph problem,
known as 
\emph{Critical 3-colorability} \cite{papadimitriou-complexitybook}:
\begin{quote}
Given a graph, is it true that the graph is not 3-colorable but 
any subgraph obtained by removing one of the edges is 3-colorable?
\end{quote}

\begin{propositionx}
For a set $\C$ of TC statements and two CQs $Q$ and $Q'$,
% and a subquery $Q'$ of $Q$, 
deciding whether $Q' = \GenC(Q)$ is $\DP$-complete.
\end{propositionx}

\begin{proof}
\emph{(Membership)} 
That the problem is in $\DP$ follows because one can verify that $Q'\equiv \GenC(Q)$ in two steps. 
For every atom in $A' \in D_{Q'}$ one has to find a statement $C \in \C$ such that $A' = T_C(D_Q)$.
This can be done with one $\NP$-oracle call.
Similarly, one has to show for every $A\in D_Q\setminus D_{Q'}$ that such a rule does not exist.
This can be done again with one call to an $\coNP$ oracle.

\emph{(Hardness)} 
We show hardness by reducing the Critical 3-colorability problem.
Assume we are given a graph $G = \tup{V,E}$ with vertices $V$ and edges $E$.
It is known that starting from $G$ one can construct a Boolean query 
$Q_G \la B_G$ where:
\[ 
 B_G = \bigwedge_{(v_i,v_j) \in E} \edge(X_i,X_j)
\]
such that $G$ is 3-colorable iff $Q_G$ evaluates to $\true$ 
over the database $D$ that contains six correct colorings for the relation $\edge$:
$\set{ \edge(red,blue), \edge (blue, red), \dots}$.
Similarly, for each edge $(v_i,v_j) \in E$ we construct a query 
$Q_G^{(i,j)} \la B_G^{(i,j)}$
where
\[
   B_G^{(i,j)} = \bigwedge_{(v_l,v_k) \in E \setminus (v_i,v_j)} \edge(X_l,X_k),
\]
such that $Q_G^{(i,j)}(D) = \true$  
iff the subgraph $G^{(i,j)} = \tup{V,E\setminus (v_i,v_j)}$ is 3-colorable.
We use this idea to construct a set of TCSs $\C$:
\begin{align*}
\C \quad=\quad  
  & \Big( \bigcup_{(v_i,v_j) \in E} \Compl(\test_{(i,j)};B_G^{(i,j)}) \Big)
  \cup \Compl(\test_G; B_G)  \\
  & \cup \Compl(\edge(X,Y);\true), 
  % \cup \Compl(\test_G; B_G)
  % \cup \Compl(\edge(X,Y);\true)
\end{align*}
where $\test_{(i,j)}$ are propositions.
Then we construct two Boolean queries:
\begin{align*}
& Q \la \Big(\bigwedge_{(v_i,v_j) \in E} \test_{(i,j)} \Big ) 
  \land \test \land D
\text{\quad and \quad}
Q' \la \Big( \bigwedge_{(v_i,v_j) \in E} \test_{(i,j)} \Big ) \land D.
\end{align*}
From the characterizing condition for completeness  
\quotes{$\theta\tpl X \in Q(T_\C(\cdb Q))$,}
we observe that each proposition $\test_{(i,j)}$ is complete iff 
$G^{(i,j)}$ subgraph is 3-colorable, and similarly
the proposition $\test$ is complete iff 
$G$ is 3-colorable.
Next, query $Q'$ is the same as $Q$ except it does not contain the $\test$ proposition.
Thus, $Q' = \GenC(Q)$ iff each $G^{(i,j)}$ is 3-colorable but 
$G$ is not.
\end{proof}

% \note{Check complexities ones again. It seems that it is not trivial.}

The fixpoint iteration algorithm computes a MCG (if it exists) 
in a number of steps linear in the size of a given query,
using a $\DP$-oracle.
% using an $\NP$-oracle.
A call to a $\DP$-oracle can be realized by two calls to an $\NP$-oracle. 
Thus, the MCG decision problem belongs to the class of problems 
$\textrm P^{\NP}$.
This is a class of problems
that can be decided in polynomial time 
by a Turing machine that uses an $\NP$-oracle.%
\footnote{An example of a complete problem for class $\textrm P^{\NP}$
provided by Krentel (1988) is the lexicographically last satisfying assignment of a Boolean formula:
Given a Boolean formula $\phi(X_1,\dots,X_n)$ the question whether in 
the in the lexicographically largest satisfying assignment of $\phi$
variable $X_n$ takes value $1$.
}

% The functional problem of MCG problem is the following: 
% Given query $Q$ and a set of TCs $\C$ return the MCG if exists;
% otherwise returns \enquote{no}.

% From the above we have that the decision problem is in $\textrm{FP}^{\NP}$;
% however it is not clear whether the functional problem also hard for this class.

% \note{Fix this below. Say that to best of our knowledge
% such DP problem cannot be characterized or drop the whole idea of explaining it.
% Think how to put it.}
% In fact, considering that decision problem is $\DP$-complete, and 
% that class $\DP$ is strictly contained in  $\textrm P^{\NP}$,
% most likely that it is not.
% %
% In the defense of the algorithm, we can say that it is not known whether the class of functional problems that corresponds to $\DP$ exists at all.
% \note{Is the following  correct???}

Hence we have the following proposition.

\begin{propositionx}
Given queries $Q$ and $Q'$, and set of TC statements $\C$,
the problem of checking whether $Q'$ is the MCG of $Q$ wrt $\C$
is in $\textrm P^{\NP}$.
\end{propositionx}

It is unknown if the problem is also $\textrm P^{\NP}$-hard.
% \note{Add a prop about the running time of the algorithm.
% Maybe this and the above one can be corollaries.}

% \note{Say that we don't know -- we cannot how hardness.}

%

\section{The Number of MCSs for Acyclic Sets of TC Statements}

\begin{theoremx}\label
Let $Q$ be a query and $\C$ an acyclic set of TCSs such that exactly $s$ relations names appear in $\C$.
% and $|Q|$ the size of query.
% Let $|\C|$ be the size of $\C$, %, $r$ the maximal arity in $\C$, 
Then the number of atoms in a MCS of $Q$ wrt $\C$ 
is in $\O(|Q|\times|C|^{s+1})$.
\end{theoremx}
% \ognjen{provide more explanations here}
\begin{proof}%[Proof (sketch)]
We analyze the number of atoms that one needs to 
add to $Q$ (in the worst case) in order to make it complete. We observe that
some atoms in $Q$ may need to be also instantiated but that 
does not affect the maximal number of atoms that one may need to consider. 
We also observe that adding more atoms may only create queries 
that are more special thus they wont be maximal any more.

Since the set of TCSs is acyclic we can create a total order on the relation names $R_1,\ldots,R_n$ in $\C$
such that  if $i<j$ then there is no TC statement with $R_j$ in the head and an $R_i$-atom in the condition.
Further, for any relation that occurs in the condition of some TC statement with $R_j$ as a head
does not have and $R_i$-atom in the condition, and so on recursively.
In other words, completeness of $R_j$-atoms does not depend on the completeness of $R_i$-atoms.

Now we can reason in the following way analyzing the atoms in $Q$.
We start with some $R_1$-atom in $Q$ (or the smallest $i$ for which $R_i$-atom exists in $Q$),
and we observe that to make each such atom complete (that is to match it with some TCS) 
we may  need (in the worst case) to introduce an additional $|\C|$ atoms in $Q$.
This is because each TC statement $\C$ is of size at most $|\C|$.
These introduced atoms may be already complete or not, but again, in the worst case we may need for each of them
to introduce again $|\C|$ new atoms from the conditions of TC statements. 
However, those new atoms are now $R_i$-atoms where $i\ge 2$.
Hence, repeating the same logic, for $R_2$ we may need to introduce 
at most $|\C|$ many $R_i$-atoms,  where $i\ge 3$, and so on. 
For $R_n$-atoms in $Q$ we cannot introduce any new atom since their TCSs
have an empty condition (otherwise there would be a loop in the dependency graph).

Hence, the total number of such steps in bounded by the number of relations in  $|\C|$ which is $s$.
Since we may need to perform this procedure for every atom in $Q$, in the worst case the resulting query can have 
$|Q| \times (|\C|+|\C|^2+\dots+|\C|^s)$ of total atoms, and so
$|Q|\times (|\C|+|\C|^2+\dots+|\C|^s) \in \O(|Q| \times|\C|^{s+1})$.

% Now we organize atoms n
% Assume $A$ be an atom in $Q$.

% To make $A$ complete we need to introduce at most $|\C|$ new atoms.
% To make complete each of those new atoms we need at most $|\C|$ new atoms for each of them; and so on. 
% Since dependency graph is acyclic the number of the above steps for $A$ is at most $s$.
% Thus, total number of atoms introduced for $A$ is
% $|\C|+|\C|^2+\dots+|\C|^s \in \O(|\C|^{s+1})$.
\end{proof}

\section{Complete Unifiers Make Queries Complete}

Applying a complete unifier to $Q$ yields a complete query:
\begin{propositionx}
If $\gamma$ is a complete unifier for a CQ $Q$ and set $\C$ of TCSs,
then $\C \models \Compl(\gamma Q)$.
\end{propositionx}
\begin{proof}

First let us recall the completeness characterization 
for minimal queries. 
Let $\C$ be a set of TCSs and  $Q(\tu) \la B$ be a minimal query then it holds that
$\C \models \Compl (Q)$ iff $\theta B \subseteq T_\C(\theta B)$.
In other words, $\C \models \Compl (Q)$ iff
for each $A\in B$ there exists a TCS $C=\Compl(A';G) \in \C$ such that
\begin{align}
& \theta A \in 
  T_{C}(\theta B)
  = \set{\beta A' \mid \beta \text{~ is a substitution ~and~~} 
          {\beta G} \subseteq \theta B}.
\label{eq:compl-char-min:1}
\end{align}

Now we return to our queries $Q$ and $\gamma Q$ from the theorem.
Wlog we assume that $\gamma Q$ is minimal as well.
%
% Assume $Q' = \alpha Q$ is returned by the algorithm.
From the definition of the unifier $\gamma$
over $\gamma Q$ and $\C$, we have that 
for each atom $A_i$ of the query there exists a 
TC statement $C_i=\Compl(A_i';G_i)$ and there are atoms $B_i\subseteq B$
such that $\gamma A_i = \gamma A_i'$ and
$\gamma G_i = \gamma B_i'$.
On the other hand,
if in equation \eqref{eq:compl-char-min:1} we set 
$\beta = \theta \gamma$ we have that

\begin{align*}
& \theta \gamma A \in 
  T_{C}(\theta \gamma B)
  = \set{\theta \gamma A' \mid\theta \gamma \text{~ is a substitution ~and~~} 
          {\theta \gamma G} \subseteq \theta \gamma B}.
\end{align*}
That is, equation 
\eqref{eq:compl-char-min:1} holds for $\gamma Q$,
so $\gamma Q$ is complete.
\end{proof}

\newpage

\section{Maximal Instantiations are Produced by Complete Unifiers}

\begin{theoremx}
Let $\C$ be a set of TCSs, $Q$ a query, and $Q'$ a complete instantiation of $Q$ wrt~$\C$.
Then there is a complete unifier $\gamma$ for $Q$ and %some $\C' \subseteq \C$ 
$\C$
such that $Q' \qcont \gamma Q$.
\end{theoremx}
\begin{proof}

We restate again characterization~(\ref{eq:compl-char-min:1})
that we established in 
the preceding proof.
% Before showing the proof we discuss the completeness characterization 
% for minimal queries. 
% Let $\C$ be a set of TCSs and  $Q \la B$ be a minimal query then 
% according to Proposition~\ref{prop-plain:minimal:set} it holds:
% $\C \models \Compl (Q)$ iff $\theta B \subseteq T_\C(\theta B)$.
Let $\C$ be a set of TCSs and  $Q(\tu) \la B$ be a minimal query then it holds that
$\C \models \Compl (Q)$ iff $\theta B \subseteq T_\C(\theta B)$.
For a minimal query $Q$, it holds that $\C \models \Compl (Q)$ iff
for each $A\in B$ there exists a TC statement $C=\Compl(A';G) \in \C$ such that
\begin{align}
& \theta A \in 
  T_{C}(\theta B)
  = \set{\beta A' \mid \beta \text{~ is a substitution ~and~~} 
          {\beta G} \subseteq \theta B}.
\label{eq:compl-char-min}
\end{align}

% \emph{(Completeness)}
%
Let $Q''$ be a complete specialization of $Q(\tX) \la A_1,\dots,A_n$  
and let $\alpha'$ be a specializing substitution such 
that $\alpha' Q \equiv Q''$.
In the following we construct a complete specialization of $Q$ that is 
at least as general as $Q$ and that is obtained by applying some 
unifier $\gamma$ on $Q$.
% and that is returned by the algorithm.
% We do so by constructing a substitution $\alpha$ such that
% \ei $\alpha Q$ is returned by the algorithm, 
% \eii $\alpha' Q \sqsubseteq \alpha Q \sqsubseteq Q$ and 
% \eiii $\alpha Q$ is complete.

Wlog we assume that $\alpha' Q$ is a minimal query.
According to characterization criteria for each atom 
$\alpha' A_i$ in $\alpha' Q$ there exists a TC statement
$C_i=\Compl(A_i';G_i)\in \C$ such that 
\[
\theta \alpha' A_i \in T_{C_i}(\theta \alpha' B)=
\set{\beta A_i' \mid \beta \text{~ is a substitution ~and~~} 
   {\beta G_i} \subseteq \theta \alpha' B}.
\]
From there it follows that 
$A_i$ and $A_i'$ are unifiable and 
that $G_i$ unifies with some set of atoms from $B$,
say $B_i \subseteq B$.
In fact, assuming that TCSs and query have different variables 
there exist substitutions $\beta_i$  such that 
$\theta \alpha' A_i = \beta_i A_i'$ and $\beta_i G_i = \theta \alpha' B_i$,
and thus $\theta \alpha' \beta_1 \cdots \beta_n$ is a unifier for all pairs.

We set that $\gamma' = \theta \alpha' \beta_1 \cdots \beta_n$
Since, pairs of are unifiable,
there exists a mgu $\gamma$ that unifies all pairs
(and that is more general than $\gamma'$, that is, $\gamma' \preceq \gamma$).
%
% Such mgu can be computed with 
% Martelli-Montanari algorithm.
% Our algorithm reuses the idea of Martelli-Montanari algorithm,
% % and creates $\gamma$ as an mgu.
% Then either \ei $\gamma Q$ is returned by the algorithm,
% or the algorithm computes an even more general specialization 
% $\gamma'' Q$ for some~$\gamma''$ 
% for which it holds $\gamma Q \sqsubseteq \gamma'' Q$.
%
Since $\gamma' \preceq \gamma$ then 
$Q' = \gamma' Q \sqsubseteq \gamma Q$.
It is also not hard to see that   $\gamma Q$ is complete.
Since we have that  $\gamma A_i = \gamma A_i'$ and
$\gamma G_i = \gamma B_i'$,
if in equation~\eqref{eq:compl-char-min} we set 
$\beta = \theta \gamma$ we have that
\[
\theta \gamma A_i \in T_{C_i}(\theta \gamma B)
\set{\theta \gamma A_i' \mid \theta \gamma \text{~ is a substitution ~and~~} 
   {\theta \gamma G_i} \subseteq \theta \gamma B}.
\]
That is \eqref{eq:compl-char-min} holds for $\gamma Q$, that is,
$\gamma Q$ is complete.

% \emph{(Soundness)}
% %
% From {(Completeness)} part we have that our algorithm 
% returns all maximal specializations. 
% Now, we need to show that it returns only complete specializations and among them only maximal ones. 

% Assume $Q' = \alpha Q$ is returned by the algorithm.
% From the construction of $Q'$ we have that 
% for each atom $A_i$ of the query there exists a 
% TC statement $C_i=\Compl(A_i';G_i)$ and atoms $B_i\subseteq B$
% such that $\alpha A_i = \alpha A_i'$ and
% $\alpha G_i = \alpha B_i'$.
% If in the equation \eqref{eq:compl-char-min} we set 
% $\beta = \theta \alpha$ we have that
% \eqref{eq:compl-char-min} holds for $\alpha Q$, i.e., 
% that is  $\alpha Q$ is complete.
% %
% Finally, the second part of the algorithm eliminates non-maximal specializations, so only the maximal ones are returned.
%
\end{proof}

\section{The k-MCS Algorithm is Correct}

% \ognjen{rewrite the proof wrt what is the paper -- align}

First we recall the algorithm (in this document listed as Algorithm~\ref{alg:kspec})
that returns all k-MCSs for a given query $Q$ and set of TCSs $\C$. 

\begin{algorithm}%[H]
 \SetAlgoNoEnd
 \LinesNumbered
 \DontPrintSemicolon
 \SetKwInOut{Input}{Input}\SetKwInOut{Output}{Output}
 \Input{a query $Q(\tpl u)\la B$, a set $\C$ of TCSs , $k \in \nn_0$}
 \Output{the set $\S$ of all $k$-MCSs of $Q$ wrt $\C$}
 % \Begin{
% let $\S$ be a set of complete specialization of size up to $k$\\
% \tcc{find complete specializations (including non-maximal)}
$\S := \emptyset$ 

% \tcc*[f]{set of complete specialization of size up to $k$} \\
%
% \For(\tcc*[f]{compute specializations up to size $k$})
% {$i\leftarrow 0$ \KwTo $k$}
{\ForEach{set $B'$ of fresh atoms of size $n+k-1$ over
  $\Sigma_\C$ }{
  \textit{construct} $Q''(\tpl u) \la B,B'$\\
  $\S := \S \cup \mathtt{MCI}_{\le n + k}(Q'',\C)$
    % $\S leftarrow \S \cup \MCI(Q'',\C)$
}

}
% \tcc*{eliminate non maximal}
\ForEach{$Q''$ in $\S$}{
% \ForEach(\tcc*[f]{eliminate non-maximal ones}){$Q''$ in $\S$}{
  \lIf{exists $Q'''\in \S$ such that $Q'' \sqsubseteq Q'''$}
  { $\S := \S \setminus \{Q''\}$}
}
\KwRet{$\S$}
% }
 \caption{Computes k-MCSs}
\end{algorithm}

\begin{theoremx}[k-MCS Algorithm]
\label{th-kspec-alg}
Let $Q$ be a CQ, $\C$ a set of TCSs,
$k$ a nonnegative integer and
$\kMCS(Q,\C)$ be the set of queries
returned by Algorithm~\ref{alg:kspec}.
Then for every query $Q'$:
% the following holds:
\[
Q' \in \kMCS(Q,\C)
\quad\text{iff}\quad 
Q' \text{~is a k-MCS of~}  Q  \text{~\wrt~}  \C.
\]
\end{theoremx}

\begin{proof}%[Proof of Theorem~\ref{th-kspec-alg}]
Let $b$ be the size of $Q$, i.e. the number of atoms in $Q$.\\

\noindent
\emph{(Soundness).}
We show that the algorithm only returns complete specializations of $Q$ of size $\le n +k$. 

\noindent An extension of $Q$ is a specialization of $Q$.
So each constructed query $Q''$ (Line 3) is a specialization of $Q$.
% of size $n+k$.
Next, the call $\mathtt{MCI}_{le n + k}(Q'',\C)$ returns 
complete specializations of $Q''$ (therefore also of $Q$) of size $\le n +k$.
Thus $\S$ can only contain such queries, and since the algorithm returns a subset of $\S$,
it can only return such queries.\\

\noindent
\emph{(Completeness).}
Let us assume that our input query $Q$ is of the form $Q(\tu) \la B$,
where $B$ is a set of $n$ atoms.
And let $Q'(\tu') \la B'$ be a complete specialization of $Q$
% of size $l$, where $n\le l \le n+k$.
of size $\le n+k$.
We show that our algorithm returns a complete specialization of $Q$ of size $\le n + k$ that is more general than $Q'$.

% More precisely,
% In the following, 
% we show that there is a query $Q''$ constructed by our algorithm Line 3
% and a complete 
% specialization $Q'''$ of $Q''$ created Line 4 such that $Q'''$ is at least as general as $Q'$.
% This guarantees that $\S$ contains at least one complete query of size $\le n + k$ that is more general than $Q'$.

% Since the algorithm returns the most general elements within $\S$,
% it is guaranteed to returns at least 
% Such $Q'''$ are returned by the subroutine $\texttt{MCI}_{\le n+k}(Q'',\C)$ that picks a maximal
% complete instantiation of size $\le n+k$ (this extends our MCI algorithm by 
% first eliminating those which size is $>n+k$ and then among them eliminating non-maximal ones).

% Since we picked $Q'$ as an arbitrary complete specialization, 
% we are guaranteed that our algorithm considers all maximal 
% complete specializations.
% Finally, our algorithm eliminates non-maximal specializations in line 5-6, thus it indeed returns all $k$-MCSs.

% Now we prove the main claim stated above:
% \begin{quote}
% For any given $Q'$ the subroutine $\texttt{MCI}_{\le n+k}$ 
% returns a query that is a complete specialization of $Q$
% and that is at least general as $Q'$.
% \end{quote}

To do so, we construct a query $Q''$ that is one of the candidate queries generated by our algorithm (Line 2) for $Q$ as input,
in such a way that $Q''$ is more general than $Q'$ (i.e. $Q' \sqsubseteq Q''$).
In other words,  $Q'$ is a complete instantiation of $Q''$.
Therefore, after the algorithm (Line 4) calls $\texttt{MCI}_{\le n+k}(Q'',\C)$,
% it will return all MCIs of $Q''$, and at least one of those
the set $\S$ will contain either $Q'$ itself or a query that is more general than $Q'$.
Together with the fact that the algorithm returns the maximal elements of $\S$ (Lines 5 to 7), this proves our claim.

The query $Q''$ is of the form $Q''(\tu) \la B,B''',B''''$, where the sets of atoms
$B'''$ and $B''''$ are defined as follows.
Since $Q'$ is a specialization of $Q$, there
exists a homomorphism $\delta \col Q \to Q'$
such that $\delta B\subseteq B'$.
Let $B''= B'\setminus \delta B$.
We observe that $|B''|<n+k-1$, since $\delta$
has to map $B$ to at least one atom in $B'$. 
We define $B'''$ as a fresh version of $B''$, that is,
a set of atoms identical to $B''$,
but where each variable is replaced with a fresh one (note that $B'''$ also has size $<n+k-1$).
To define $B''''$, we select an arbitrary atom $A$ in $B$,
and construct $B''''$ with enough repeated fresh version of $A$ so that
$|B|+|B'''|+|B''''|= n+n+k-1$.

To conclude the proof, we show that $Q' \sqsubseteq Q''$.
First, we observe that $B''''$ is \enquote{redundant} in $Q''$, meaning that dropping 
$B''''$ yields an equivalent query.
So in order to show that $Q' \sqsubseteq Q''$, 
it is sufficient to show that there is a homomorphism from $(B, B''')$ to $B'$.
Such a homomorphism can be constructed by extending $\delta$ (which, as a reminder, maps $B$ to $B'$)
with a substitution $\delta'$ that maps $B'''$ to $B'$ and is compatible with $\delta$.
Trivially, from the construction of $B'''$, such a substitution exists:
since $B'''$ is constructed out of $B'' \subseteq B'$ by renaming variables with fresh ones,
$\delta'$ can be defined as the inverse of this renaming.

% because starting from $\delta$
% one can construct a homomorphism $\delta '$ from $Q''$ to $Q'$
% such that $\delta '$ and $\delta$ match on the variables from $Q$
% and then extend it cover the variables in $B''$ and $B'''$
% so the image of $B''$ is exactly $B'$ and of $B'''$ is $A$.

% Now recall that $Q''(\tu) \la B,B''',B''''$ is one of the queries created by our algorithm (Line 2), for $Q$ as input. 

% However, our algorithm creates exactly such queries in line 2.

% Our algorithm directly guesses a fresh extension of size $k$.
% Thus, to match our algorithm, we do the following. 
% Let $A$ be an atom in $B$. Then we construct $B''''$
% by repeating fresh version of $A$ such that
% $|B|+|B'''|+|B''''|= n+n+k-1$.
% Now let $B''$ be an atom from $B'$ and let use construct $B'''$
% that contains as many of such as (each with fresh variables)
% such $|B|+|B''|+|B'''|=n+k$.
% We construct the query $Q''(\tu) \la B,B''',B''''$.
% We observe that $B'''$ is \enquote{redundant} in $Q''$, that is, dropping 
% $B'''$ one gets an equivalent query.
% However, our algorithm creates exactly such queries in line 2.
% First it is not hard to 
% check $Q''$ is an extension of $Q$
% such that size of $Q''$ is $|Q|+k$, 

% Now we observe that $Q' \sqsubseteq Q''$, 

\end{proof}

%%% Local Variables:
%%% mode: latex
%%% TeX-master: t
%%% End:

% Finally, for finiteness over dense time, all negative results (i.e.~non-finiteness) are illustrated in the article, and all positive results (i.e.~finiteness) follow from the definitions of the inductive representations (and the fact that they are correct).

% % \sloppy
% % % \setcounter{tocdepth}{3}
% % % \makeatletter
% % % \renewcommand*\l@author[2]{}
% % % \renewcommand*\l@title[2]{}
% % % \makeatletter
% % \begingroup
% % \let\clearpage\relax
% % \tableofcontents
% % % \endgroup
% % \renewcommand{\tableofcontents}{
% %   \section*{\contentsname}
% %   \@starttoc{toc}
% % }

% % \section*{Table of Content}
% % % \startcontents[chapter]
% % \printcontents[chapter]{l}{0}{\setcounter{tocdepth}{4}}
% \end{appendix}

\end{document}